\documentclass[11pt]{article}
\usepackage{jheppub}

\usepackage[table]{xcolor}
\usepackage{color}
\usepackage{graphicx,amsmath,amsthm,amssymb,array,tikz-cd}
\usepackage{minibox}
\usepackage{ccaption}

  \captionnamefont{\bfseries}
  \captiontitlefont{\small\sffamily}
  \captiondelim{: }
  \hangcaption

\newcommand{\notimplies}{%
  \mathrel{{\ooalign{\hidewidth$\not\phantom{=}$\hidewidth\cr$\implies$}}}}

\newtheorem{theorem}{Theorem}
\DeclareMathOperator{\im}{im}
\definecolor{DarkGreen}{RGB}{0,130,0}

\DeclareMathOperator{\Tr}{Tr}
\newcommand{\bdy}{\mathcal{B}}
\newcommand{\bulk}{\mathcal{M}}
\newcommand{\regA}{\mathcal{A}}
\newcommand{\regAc}{{\regA^c}}
\newcommand{\rhoA}{{\rho_{\regA}}}

\newcommand{\entsurf}{\partial \regA}
\newcommand{\homsurfA}{\mathcal{R}_\regA}
\newcommand{\extr}{\mathcal{E}}
\newcommand{\fixM}{\mathbf{e}}
\newcommand{\intM}{\mathbf{r}}

\newcommand{\uE}{u}
\newcommand{\disk}{\mathcal{D}}

\newcommand{\extrhom}{e}
\newcommand{\reghom}{a}

\newcommand{\cutout}[2]{#1- #2}
\newcommand{\bdyc}{\cutout{\bdy}{\entsurf}}
\newcommand{\bdycCover}{\cutout{\tilde{\bdy}_q}{\entsurf}}

\newcommand{\bulkcfix}[1]{\cutout{\bulk_{#1}}{\fixM_{#1}}}
\newcommand{\bulkcCover}[1]{\cutout{\tilde{\bulk}_{#1}}{\fixM_{#1}}}

\definecolor{darkgreen}{rgb}{0,0.45,0.2}

\title{Topological aspects of generalized gravitational entropy}

\author[a]{Felix M. Haehl\,}
\author[b]{\!, Thomas Hartman\,}
\author[c]{\!, Donald Marolf\,}
\author[a]{\!, Henry Maxfield\,}
\author[a]{\!, Mukund Rangamani\,}
\affiliation[\,a]{Centre for Particle Theory \& Department of Mathematical Sciences,\\
                     Durham University, South Road, Durham DH1 3LE, UK.}
\affiliation[\,b]{Department of Physics,\\
                     Cornell University, Ithaca, New York 14853, USA.}
\affiliation[\ c]{Department of Physics,\\
                     University of California, Santa Barbara, Santa Barbara, CA 93106, USA.}
\emailAdd{f.m.haehl@gmail.com}
\emailAdd{hartman@cornell.edu}
\emailAdd{marolf@physics.ucsb.edu}
\emailAdd{h.d.maxfield@durham.ac.uk}
\emailAdd{mukund.rangamani@durham.ac.uk}

\abstract{The holographic prescription for computing entanglement entropy requires that the bulk extremal surface, whose area encodes the amount of entanglement, satisfies a homology constraint. Usually this is stated as the requirement of a (spacelike) interpolating surface that connects the region of interest and the extremal surface. We investigate to what extent this constraint is upheld by the generalized gravitational entropy argument, which relies on constructing replica symmetric
$q$-fold covering spaces of the bulk, branched at the extremal surface. We prove (at the level of topology) that the putative extremal surface satisfies the homology constraint if and only if the corresponding branched cover can be constructed for every positive integer
$q$. We give simple examples to show that homology can be violated if the cover exists for some values $q$ but not others, along with some other issues.
}

\begin{document}
	\begin{flushright} \small{DCPT-14/75} \end{flushright}
	
\maketitle
\flushbottom

\section{Introduction}
\label{sec:introduction}

Holography provides an intriguing connection between quantum information and geometry,  in part inspired by the geometrization of quantum entanglement by the Ryu-Takayanagi (RT) proposal \cite{Ryu:2006bv,Ryu:2006ef} and its covariant generalization \cite{Hubeny:2007xt} (HRT). These proposals identify a particular bulk codimension-2 surface as the geometric encoder of the entanglement structure in the dual field theory.  Whilst the original proposals owed their origins to analogies with black hole entropy and covariant entropy bounds, we now have a derivation of the RT formula (in situations with time reflection symmetry) from a gravitational path integral courtesy of Lewkowycz and Maldacena (LM) \cite{Lewkowycz:2013nqa}.

This may at first seem like a complete story for the time-reflection symmetric states of a holographic QFT. Although it requires certain assumptions, the LM construction gives a first principles derivation of how the minimal surface of the RT proposal comes about from a gravitational path integral. But while the LM construction captures the dynamical part of the RT conjecture, to our knowledge it has yet to be established whether the quantum gravity path integral employed by LM is cognizant of the topological constraints that must be imposed on the RT surface. This is the primary question that will concern us in this paper.

To appreciate the issues involved, recall that RT (HRT) proposes that the holographic entanglement entropy of a given region $\regA$ of the field theory is given by the area of a minimal (extremal) bulk surface $\extr$ anchored on the boundary $\entsurf$ of $\regA$, with $\regA$, $\entsurf$ considered to lie in the boundary of the bulk spacetime.
The extremal surface $\extr$ is required to be homologous to the region $\regA$ in question \cite{Fursaev:2006ih,Headrick:2007km}. The cleanest phrasing of this statement to our knowledge appears in
\cite{Headrick:2013zda}, \cite{Headrick:2014cta}  for the RT and HRT proposals respectively. One requires that there be a spacelike codimension-1 interpolating homology surface $\homsurfA$ whose only boundaries are $\extr$ and $\regA$.
Without this constraint, the holographic formula would be at odds with known quantum mechanical properties of entanglement entropy, such as strong subadditivity; the topological prescription has been motivated as a natural, but \textit{ad hoc} way of ensuring consistency. The simplest rationale for the homology constraint is that, in its absence, a subsystem $\regA$ and its complement $\regA^c$ can end up having the same entanglement even when the state of the system is impure (e.g., in a black hole geometry).

One can intuitively motivate this picture by realizing that one must introduce a cut for fields along
$\regA$ when computing matrix elements of the reduced density matrix $\rhoA$ via the path integral.  In particular, the replica construction for computing R\'enyi entropies requires that the operators of the QFT are cyclically permuted when one crosses the cut. Since local QFT operators are the boundary values of bulk fields, one expects the cut to extend into the bulk as well. The boundary of the bulk branching surface $\homsurfA$ is formed from $\extr$ and $\regA$ so that the homology constraint is satisfied; see Fig.~\ref{fig:replica} for an illustration. Existence of the homology surface $\homsurfA$ ensures that fields are appropriately branched and one can treat the bulk geometry itself as a branched cover over some fundamental domain.

Based on the above arguments, one might imagine
 that the converse holds true trivially at the level of topology.
 That is, given a replica-symmetric bulk saddle with the correct boundary conditions, might we be guaranteed to find a codimension-2 defect and a codimension-1 interpolating surface implementing the homology constraint?
 Surprisingly, this turns out not to be true. One can construct bulk geometries with the requisite replica symmetry but which nevertheless do not admit an appropriate branching surface $\homsurfA$. While it is plausible that such geometries are never dominant in the bulk path integral, their presence begs the question ``under what conditions does the LM construction give rise to the homology constraint?'' We will address this question in some detail below.  A simplified version of our final statement is that, as long as one has a family of geometries parameterized by a real parameter $q$ which for all positive integer values of $q$ admits a branched cover description, then we recover the homology constraint from the LM construction.

 Before beginning, we pause to dispel a potential confusion.  The reader may well ask whether the above condition coincides with the assumptions actually made by Lewkowycz and Maldacena in their original paper \cite{Lewkowycz:2013nqa}.  The potential confusion arises from the fact that the phrasing of the assumptions in \cite{Lewkowycz:2013nqa} is subject to at least two different interpretations, which are naturally termed `global' and `local'. Indeed, the present authors do not agree among themselves as to which interpretation best fits the words written in \cite{Lewkowycz:2013nqa}.   Under the global interpretation, LM assumes that the replica geometries may be analytically continued to $q \sim 1$ with global topology outside the conical singularity given by a trivial
 ${\bf S}^1$ bundle ${\bf S}^1 \times X$ for some $X$.  This global assumption coincides with the intuitive discussion above and immediately implies the homology constraint, as the homology surface is given by any global section of this trivial bundle.

 However, the global assumption is rather stronger than one might like.  As mentioned above, there is no a-priori reason why a ($q$-fold quotient of a) given replica-symmetric solution at integer $q$ should have this structure.  Furthermore, as may be readily seen by noting that there is no obstacle to constructing bulk solutions with small conical singularities on homology-violating surfaces, the global product structure is not required by the assumption that the $q \to 1$ limit of the bulk solutions be described by small conical singularities.  Such considerations motivates us to focus on the alternate and weaker local interpretation of the LM assumptions, which does not determine the global structure  and imposes the ${\bf S}^1 \times X$ structure only locally along the conical singularity.  The implications are two-fold.  First, as is necessary to allow LM to work in topologically nontrivial examples (such as computing the entropy of an interval on a torus), the ${\bf S}^1$ fiber need not be defined far from the singularity.  Second, even the region near-but-outside the singularity is allowed to be a non-trivial ${\bf S}^1$ bundle over some $X$. Our main result can now be restated as saying that this local interpretation of the LM assumptions \textit{also} implies the homology constraint so long as the family of geometries exists at all $q\geq1$ and is an appropriate $\mathbb{Z}_q$ quotient of a smooth $q$-fold replica at all integer $q$.

The outline of the paper is as follows: in \S\ref{sec:review} we give a rapid overview of the concepts we need from QFT and holography vis-\`a-vis entanglement considerations. We also formulate the above issue concerning the homology constraint in detail. In \S\ref{sec:crosscap} we present counterexamples to the na\"ive intuition that the homology constraint is an automatic consequence of existence of a branched cover. In \S\ref{sec:consistency} we formulate an essential topological consistency condition for branched covers as constructed by LM. This condition is a statement about consistency of the field theory replica trick with the existence of bulk branched covers at all integer values of $q$. We then illustrate and prove in \S\ref{sec:homology} that this topological consistency condition is in fact equivalent to the homology constraint on the RT prescription. We conclude with a discussion in \S\ref{sec:discussion}. Further technical details and a review of algebraic topology relevant for our proof can be found in the appendices.

\section{Review of holographic entanglement}
\label{sec:review}

To set the stage for our discussion, let us quickly review the salient features of the RT/HRT and LM constructions that will play a role below. Consider a holographic $d$-dimensional QFT with a bulk gravity dual in asymptotically $d+1$ dimensional AdS spacetime. We will limit our discussion to situations where the QFT is planar, strongly coupled, and has a suitable gap in its spectrum, so that we can regard the bulk as a two derivative effective field theory which we take to be Einstein-Hilbert gravity coupled to matter degrees of freedom.\footnote{ The discussion generalize​s​ to other planar QFT​s​ dual to (classical) higher derivative  gravitational theories following \cite{Dong:2013qoa, Camps:2013zua}​, for which we would be evaluating some other local functional on a codimension-2 surface $\extr$. As we primarily focus in on topological aspects, much of what we describe later will go through ​unmodified​.}

\subsection{The RT and LM constructions}
\label{sec:LMRT}

The QFT resides on a background geometry $\bdy$  foliated by  Cauchy surfaces
$\Sigma_t$.\footnote{\label{foot:tsym} We often use language appropriate for Lorentzian spacetimes despite focusing on Euclidean geometries. We will be interested in surfaces $\Sigma_{t=0}$ which sit at a moment of ${\mathbb Z}_2$ time-reflection symmetry (allowing thereby translation between the two cases by a suitable Wick rotation).}  The dual gravitational background to this field theory is $\bulk$ and has $\bdy$ as its conformal boundary. Specifically, we will take $\bulk$ (and all other bulk regions) to denote the conformal compactification, so that $\bulk$ is a manifold with boundary: $\partial \bulk = \bdy$. See Table \ref{tab:notation} for an overview of our notation.

\begin{table}[h]
\centering
\begin{tabular}{| c | l | c |}
\hline
\multicolumn{3}{|c|}{{\cellcolor{gray!13} Boundary regions}}\\
\hline\hline
 Symbol & Description & Dimension \\
 \hline
 $\bdy$ & full boundary manifold & $d$ \\
 $\Sigma_t$ & fixed-time slice & $d-1$ \\
 $ \regA$ & subregion of $\Sigma_t$ & $d-1$ \\
 $\entsurf$ & entangling surface & $d-2$\\
 $ \tilde{\bdy}_q$ & $q$-fold branched cover of $\bdy$ used in the replica trick & $d$ \\
\hline\hline
\multicolumn{3}{|c|}{{\cellcolor{gray!13} Bulk regions}}\\
\hline\hline
 Symbol & Description & Dimension \\
 \hline
 $\bulk$ & full bulk manifold (with $\partial \bulk = \bdy$) & $d+1$ \\
 $\extr$ & RT minimal surface (with $\partial \extr = \partial \regA$) & $d-1$ \\
 $\homsurfA$ & homology surface interpolating between $\extr$ and $\regA$ & $d$ \\
 $\tilde{\bulk}_q$ & smooth bulk replica manifold (with singular $\partial \tilde{\bulk}_q = \tilde{\bdy}_q$) & $d+1$ \\
 $\bulk_q$ & fundamental domain of $\tilde{\bulk}_q$ when $\tilde{\bulk}_q$ is a branched cover (i.e.\ $\tilde{\bulk}_q/\mathbb{Z}_q$)& $d+1$\\
 $\fixM_q$ & branching surface of the branched cover $\tilde{\bulk}_q \to \bulk_q$ & varies \\
 $\intM_q$ & homology surface interpolating between $\fixM_q$ and $\regA$ (if it exists) & $d$ \\
\hline
\end{tabular}
\caption{Definition of boundary and bulk regions that we consider in the course of our discussion.}
\label{tab:notation}
\end{table}

We are interested in computing the entanglement entropy in such a holographic QFT for a region $\regA \subset \Sigma_t$ of the boundary geometry. The HRT prescription for computing holographic entanglement entropy requires us to consider a bulk codimension-2 extremal  surface  $\extr \subset \bulk$ anchored on $\entsurf$ and takes the area of $\extr$ to give the boundary entaglement, i.e.,
\begin{equation}
S_\regA = \frac{\text{Area}(\extr)}{4 \,G_N}\,, \qquad
\partial \extr = \extr \cap \bdy = \entsurf
\,.
\label{eq:SAbulk}
\end{equation}
 Crucially the bulk extremal surface is required to satisfy a {\em homology constraint}  originally motivated in  \cite{Headrick:2007km}.  Usually this is stated as the requirement that the extremal surface must be homologous to the region $\regA$. More precisely, following e.g., \cite{Headrick:2013zda,Headrick:2014cta}  we will take this to mean that there exists a bulk codimension-1 spacelike surface $\homsurfA$ which interpolates between the extremal surface and the boundary region of interest. To wit,
\begin{equation}
\exists \;\homsurfA \subset \bulk : \quad\partial \homsurfA = \extr \cup \regA \,.
\label{eq:radef}
\end{equation}	
In the RT/HRT constructions, taking $\regA$ to be spacelike is an additional restriction on the allowed minimal/extremal surfaces, though it is naturally incorporated in the maximin proposal of  \cite{Wall:2012uf}.
We will have much more to say about this constraint below.

We need one more ingredient to make contact with the LM path integral derivation.  This ingredient is the replica trick for computing powers of the reduced density matrix $\rho_\regA$ whose von Neumann entropy is the entanglement entropy under discussion.  We define
\begin{equation}
\rho_\regA = {\rm Tr}_{\regAc} \rho \,, \qquad
S_\regA = -\Tr\left(\rho_\regA \log \rho_\regA\right) \,,
\label{eq:SA}
\end{equation}	
where $\rho$ is the total density matrix on $\Sigma_t = \regA \cup \regAc$. After setting $t=0$ as in footnote \ref{foot:tsym} and passing to Euclidean signature, $\rho$ can be viewed as a state prepared by some path integral over $\bdy$ (now denoting a Euclidean boundary) cut along $\Sigma_{t=0}$ where boundary conditions are imposed on fields at $\Sigma_{t=0}$ to compute particular matrix elements of $\rho$. Up to a normalizing factor, the trace over $\regAc$ to obtain $\rho_\regA$ is implemented by sewing up the part of the cut along $\regAc$, leaving $\bdy$ with a cut only along $\regA$. The replica construction in the QFT then proceeds by sewing $q$ copies together cyclically along the cuts at ${\cal A}$ to construct a singular manifold $\tilde{\bdy}_q$ whose partition function computes $\text{Tr} (\rho_\regA^q)$. This then allows us to recover the $q^{\rm th}$ R\'enyi entropy of the QFT via:
\begin{equation}
S_\regA^{(q)}
= \frac{1}{1-q} \, \log \text{Tr} (\rho_\regA^q) = \frac{1}{1-q} \log \frac{Z_q}{Z_1^q} \,,
\label{eq:qrenyi}
\end{equation}	
where $Z_q$ is the partition function of the QFT on $\tilde{\bdy}_q$ and $Z_1$ that on $\bdy_1 \equiv \bdy$. The entanglement entropy $S_\regA$ of \eqref{eq:SA} is recovered in the limit $q\to 1$.

The LM construction first implements this computation of R\'enyi entropies holographically by extending the replica trick into the bulk.  It then extracts the entanglement entropy as above by giving a geometric implementation of the continuation to non-integer $q$.

To compute $Z_q$, one proceeds by obtaining a bulk manifold $\tilde{\bulk}_q$ with boundary $\tilde{\bdy}_q$ (for some explicit examples see \cite{Faulkner:2013yia,Barrella:2013wja}); as always, the partition function is simply given by the on-shell gravitational action computed on this geometry. This bulk computation follows the usual rules of Euclidean quantum gravity and the LM saddle point analysis remains valid when we analytically continue  $q\in  {\mathbb Z}_+ \mapsto q \in {\mathbb R}$, as long as $(q-1)\, \ell_{AdS}/\ell_{Planck} \gg 1$.

To get from the replica spacetimes $\tilde{\bulk}_q$ of LM to the RT minimal surface one proceeds as follows. Since the boundary geometry $\tilde{\bdy}_q$ is a $q$-fold cover over $\bdy$ branched at $\entsurf$ with cyclic ${\mathbb Z}_q$ symmetry, we can restrict attention to a single `fundamental domain' by focusing on the quotient spacetime $\tilde{\bdy}_q/{\mathbb Z}_q$; this is just a copy of $\bdy$ itself. Assuming the bulk saddle point geometry $\tilde{\bulk}_q$ to respect replica symmetry,\footnote{
The ${\mathbb Z}_2$  time-reflection symmetry about $t=0$ of the state $\rho$ intertwines with the cyclic
${\mathbb Z}_q$ symmetry of the replica construction, to give a larger dihedral symmetry group ${\mathbb D}_q$; see \cite{Headrick:2012fk} for its relevance in computing R\'enyi entropies. We refrain from utilizing the full dihedral symmetry, allowing for the possibility that the LM construction gives a surface that does not lie at $t=0$ in the bulk. Therefore, in what follows, replica symmetry will always refer to the cyclic ${\mathbb Z}_q$ group.} we may similarly consider the bulk quotient $\bulk_q=\tilde{\bulk}_q/{\mathbb Z}_q$. LM focus on the case where the action of ${\mathbb Z}_q$ on $\tilde{\bulk}_q$ has a codimension-2 fixed point set $\fixM_q$ with boundary $\entsurf$.  This $\fixM_q$ is to be identified as the progenitor of the extremal surface $\extr$.  It is assumed to result in a conical defect of angle $\frac{2\pi}{q}$ in $\bulk_q$. The desired QFT partition function $Z_q$ on $\tilde{\bdy}_q$ is then $q$ times the bulk action on $\bulk_q$, computed without a contribution from the conical defect.

The point of considering this quotient is that it allows continuation to arbitrary real values of $q$. The protocol is to find a geometry $\bulk_q$ with boundary $\bdy$ and from which a `singular' codimension-2 surface $\fixM_q$ ending at $\entsurf$ has been excised.  One then imposes as a further boundary condition that $\fixM_q$ is a conical defect of opening angle $\frac{2\pi}{q}$. The geometry is fixed by minimizing the action subject to this requirement, with no contribution to the action from the singularity.  In the $q\to 1$ limit we require $\bulk_q \to \bulk$.  The defect surface $\fixM _q$ then becomes the minimal area surface $\extr$ in the Euclidean geometry. One can furthermore argue that the contribution to $Z_q$ localizes on this surface, giving a correction to the action proportional to the area, in such a way that the area of the extremal surface computes the entanglement entropy.

\subsection{A question of homology}
\label{sec:LMtopology}

The LM construction shows that the computation of R\'enyi entropies for arbitrary positive real $q$ can be performed by finding geometries $\bulk_q$ with boundary $\bdy$, but with a conical defect of angle $\frac{2\pi}{q}$. At integer values of $q$, via $\bulk_q = \tilde{\bulk}_q/\mathbb{Z}_q$ this should be related to a nonsingular replicated bulk $\tilde{\bulk}_q$ whose boundary is the replica $\tilde{\bdy}_q$ (on which this $\mathbb{Z}_q$ acts as the replica symmetry). When this is the case, we say that $\bulk_q$ lifts to a $q$-fold branched cover. Of primary interest to us is the relationship between the homology constraint on the one hand, and this lifting of the singular bulk geometry to a nonsingular replicated bulk on the other.

The bulk conical defect, coming from the fixed point set of the replica symmetry, is the codimension-2 surface $\fixM_q$ anchored on the boundary at $\entsurf$. On the boundary $\tilde{\bdy}_q$, upon traversing a small loop around $\entsurf$, one passes through $\regA$, and goes from one copy of $\bdy$ to the next. It is tempting to imagine a natural picture of the full geometry that arises from continuing this reasoning into the bulk as follows: At the level of topology, $\tilde{\bulk}_q$ is formed by sewing together $q$ copies of $\bulk_q$  along some $\fixM_q$ such that traversing a small loop around $\fixM_q$ also results in a change of sheet in the cover $\tilde{\bulk}_q$ of $\bulk_q$.  We shall investigate the correctness of this picture below.

This picture is straightforward in cases where there is a codimension-1 interpolating surface $\intM_q$ bounded by $\fixM_q \,\cup\, \regA$: in other words, if the homology constraint is satisfied by the conical defect. Cutting along $\intM_q$, and gluing together $q$ copies cyclically, just as on the boundary, builds the covering space $\tilde{\bulk}_q$ at the level of topology. Passing through $\intM_q$ will cause a change of sheet in the cover, just as passing through $\regA$ changes sheets on the boundary. So any $\extr$ obeying  the homology constraint\footnote{The natural homology constraint to assume is the one in Lorentz signature.  We will take the Lorentz signature spacetime to be time-orientable.  It then admits a nowhere-vanishing vector field along which $\homsurfA$ can be deformed until it lies in the moment of time symmetry.  It follows that the homology constraint also holds in Euclidean signature.  For future reference, we also note that Lorentz-signature time-orientability makes $\homsurfA$ two-sided, i.e., it has a continuous unit normal. In case of an orientable bulk, two-sidedness is equivalent to orientability of $\homsurfA$.} naturally leads to a family of bulk saddles $\tilde{\bulk}_q$ obeying the correct boundary conditions.
We may say that $\bulk$ lifts to $\tilde{\bulk}_q$, with $\extr$ lifting to $\fixM_q$.\footnote{When
all surfaces of interest lie on the $t=0$ slice, $\homsurfA$ is uniquely determined by $\extr$ and $\regA$ as long as that slice  has no closed, boundaryless components. More generally, $\homsurfA$ is unique if
$H_{d}(\bulk) =0$, which we expect to hold in most physical situations (see, for example \cite{Witten:1999xp}).
If $H_d{(\bulk}) \neq 0$ then there is an ambiguity, though it will not make a difference to the entanglement entropy in the classical limit (but it will matter for R\'enyi entropies and for quantum corrections).
}
See Fig.~\ref{fig:replica} for an illustration of this scenario.
\begin{figure}
\setlength{\unitlength}{0.1\columnwidth}
\centerline{\includegraphics[width=\textwidth]{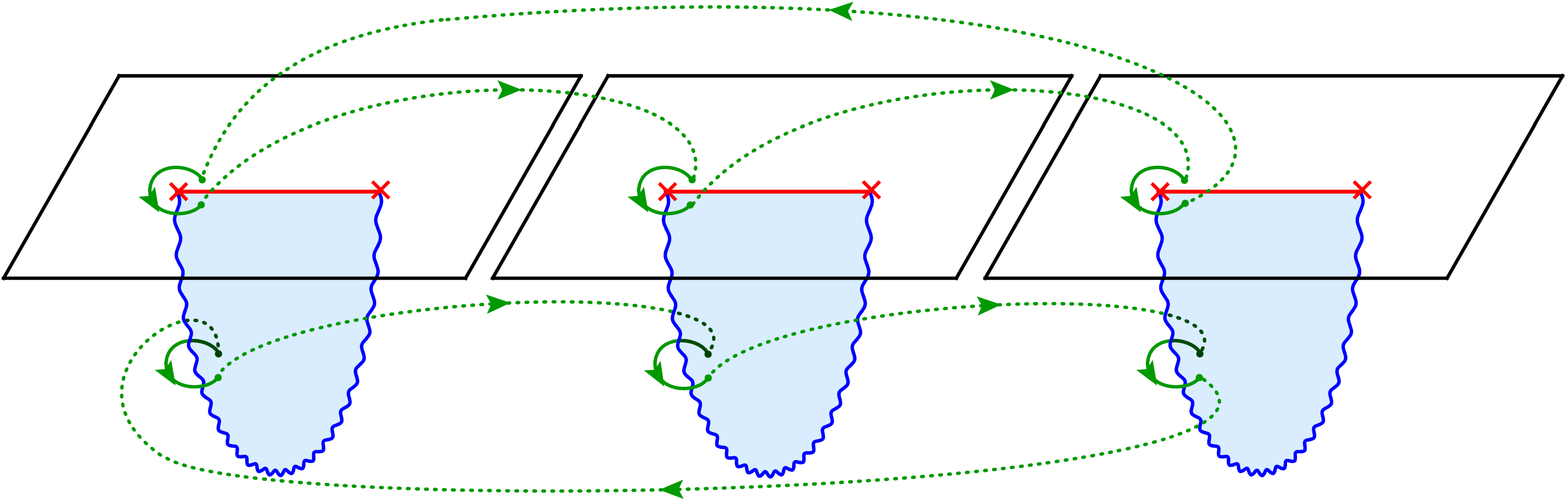}}
\begin{picture}(0.3,0.4)(0,0)
\put(1.8,2.5){\minibox{{\color{red}$\regA$}}}
\put(2.35,1){\minibox{{\color{blue}$\fixM_3$}}}
\put(1.7,1.15){\minibox{{\color{cyan}$\intM_3$}}}
\put(.75,2.85){\minibox{$\bdy$}}
\end{picture}
\caption{Replica construction in the boundary and bulk for $q=3$. The replica symmetric $q$ copies of the field theory on $\bdy$, form a $q$-fold branched cover $\tilde{\bdy}_q$ which fixes the asymptotic data for the bulk problem.
The bulk covering spacetime $\tilde{\bulk}_q$ has a ${\mathbb Z}_q$ symmetry with fixed point locus $\fixM_q$  (shown as the wavy lines) anchored on $\entsurf$. Typically one also encounters via this construction a bulk interpolating surface  $\intM_q$ (the light blue branching surface) in the bulk whose boundaries are $\fixM_q$ and $\regA$.
Conventional intuition dictates that the bulk spacetimes are all  covers over a single fundamental domain (one of the components in the picture) branched over the codimension-1 surface $\intM_q$. Passing through this surface cycles through the sheets of the bulk in a fashion identical to passage through $\regA$.
The homology condition posits that such an $\intM_q$ exists. We argue that this picture is accurate as long as we are suitably careful with the notion of allowed branched covers.  As $q \rightarrow 1$, $\intM_q \rightarrow \homsurfA$ and $\fixM_q \rightarrow \extr$.
}
\label{fig:replica}
\end{figure}
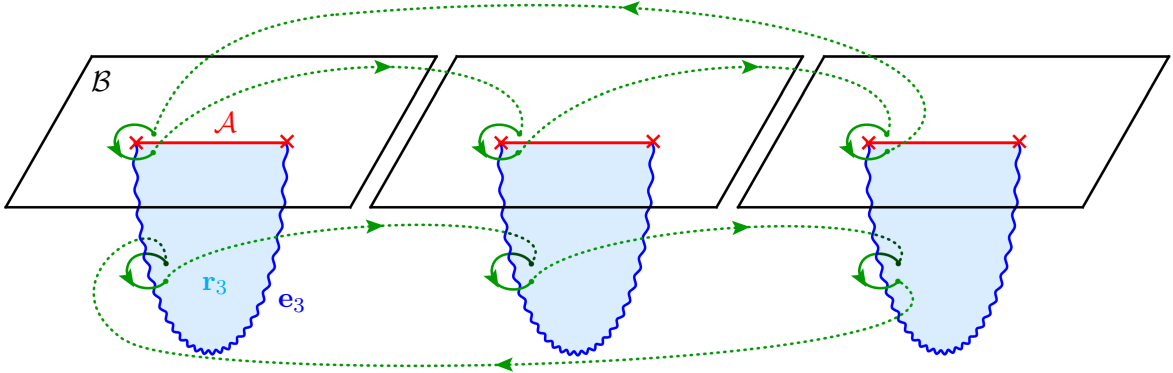

However, the converse, or at least the strongest converse one might propose, fails to be true. More specifically, the following three statements will be demonstrated in sections \S\ref{sec:crosscap}, \S\ref{sec:consistency}, and \S\ref{sec:homology}.
\begin{enumerate}
\item[(i).] The existence of a two-sided interpolating codimension-1 surface $\intM_q$ implies that $\bulk_q$ lifts to a branched cover (homology $\implies$ lift). This is what we have informally argued above.
\item[(ii).] There are branched covers with the correct boundary conditions which cannot be realized in this way.  That is, for given $q \in \mathbb{Z}_+$, there can exist an $\bulk_q$ formed from a quotient of a branched cover which does not admit an interpolating surface $\intM_q$ between $\regA$ and the fixed point set $\fixM_q$ (lift $\notimplies$ homology).
\item[(iii).] However, given a continuous family of bulk geometries $\bulk_q$ parameterized by real $q$ which for every
$q\in {\mathbb Z}_+$ lifts to a $q$-fold branched cover, we will show that each $\bulk_q$ admits an interpolating surface $\intM_q$.  Taking $q \rightarrow 1$ then shows that $\extr$ satisfies the homology constraint as desired (lift $\forall \; q \implies $ homology).
\end{enumerate}

We also note that, on top of this, the fixed point set arising from a $\mathbb{Z}_q$ quotient may give rise to something other than a $\frac{2\pi}{q}$ conical defect. It is possible to generate fixed point sets with the wrong codimension, and also to engineer situations wherein the codimension-2 fixed point set has an incorrect defect angle. In \S\ref{sec:crosscap} we give examples where both these scenarios can be realized.

Below we generally confine ourselves to topological arguments. In particular, we refrain from employing dynamical information from the path integral to constrain $\bulk_q$. This is in part due to the fact that classification of all replica invariant saddles with given boundary geometry is a notoriously hard problem (even for $d=2$).\footnote{ The kinematical aspect of our analysis is reminiscent of the first attempt to prove the RT proposal in \cite{Fursaev:2006ih} where the homology condition was introduced. But while \cite{Fursaev:2006ih} argued that the bulk should be a branched cover satisfying homology, we would instead like to ascertain the conditions under which the homology condition becomes automatic. We thank Matt Headrick for a discussion on this issue.}
In the same spirit, the notion of replica symmetry should always be understood in a topological sense.  Formally, the bulk will be a multi-sheeted surface -- we refrain from specifying the detailed geometry on the sheets, but do keep track of what the boundary conditions of the gravitational path integral imply for moving across the sheets (see \S\ref{sec:consistency}).

\section{Exemplifying replica symmetric homology violation}
\label{sec:crosscap}

We begin by considering some simple examples, focusing only on the topology, which illustrates some of the unexpected features that may be encountered.

\subsection{A torus with a crosscap}

Suppose we wish to compute the entropy of the thermal state of a two-dimensional field theory living on a circle (or equivalently the entanglement entropy between the two halves of the thermofield double). The state is defined by a path integral over a cylinder of length $\beta=1/T$. The $q^{\rm th}$ R\'enyi entropy can be computed by the partition function on $q$ copies of the cylinder glued cyclically, which gives a torus $\tilde{\bdy}_q$ of length $q\beta$. The replica symmetry is implemented by rotating through $\beta$ in the Euclidean time direction. The region $\regA$ of interest becomes a spatial circle at fixed Euclidean time, with empty boundary.

To compute the partition function holographically, one must choose how the torus is to be `filled in' to obtain a bulk $\tilde{\bulk}_q$. One option is to fill the spatial circle with a disk, like thermal AdS, in which case there are no fixed points under the replica symmetry, so the quotient $\bulk_q$ is smooth, and the entanglement entropy obtained would vanish. The region $\regA$ is contractible in the bulk, so this is consistent with the homology constraint.

If, on the other hand, the Euclidean time circle is filled in with a disk, as in the BTZ black hole geometry, $\regA$ is not contractible. But the centre of the disk is fixed under the rotation implementing replica symmetry.  So it gives rise to a circle $\fixM_q$ of fixed points and a conical defect in the quotient $\bulk_q$.  This defect becomes the bifurcation circle of the event horizon in the $q\to1$ limit. There is an obvious interpolating surface joining it to $\regA$ on the boundary, so again the homology constraint is obeyed.

However, ignoring for now the equations of motion,
there are more topologies that might in principle be allowed. One is to fill the Euclidean time circle not with a disk, but using a cross-cap. (In what follows, the spatial circle will play little role, so we focus on some 2-dimensional slice corresponding to a point on this circle.) This means that we first fill it with an annulus, taking the outside edge of the annulus to be the boundary, but we then identify antipodal points of the inside edge so that passing through the inside edge causes one to jump to the point directly opposite. The replica symmetry can be implemented by rotation of the cross-cap in the obvious manner. This gives a space $\tilde{\bulk}_q$ with the topology of a M\"obius strip  times the spatial circle. The (single) edge of the M\"obius strip is the boundary Euclidean time circle.

\begin{figure}
\setlength{\unitlength}{0.1\columnwidth}
\centerline{\includegraphics[width=.91\textwidth]{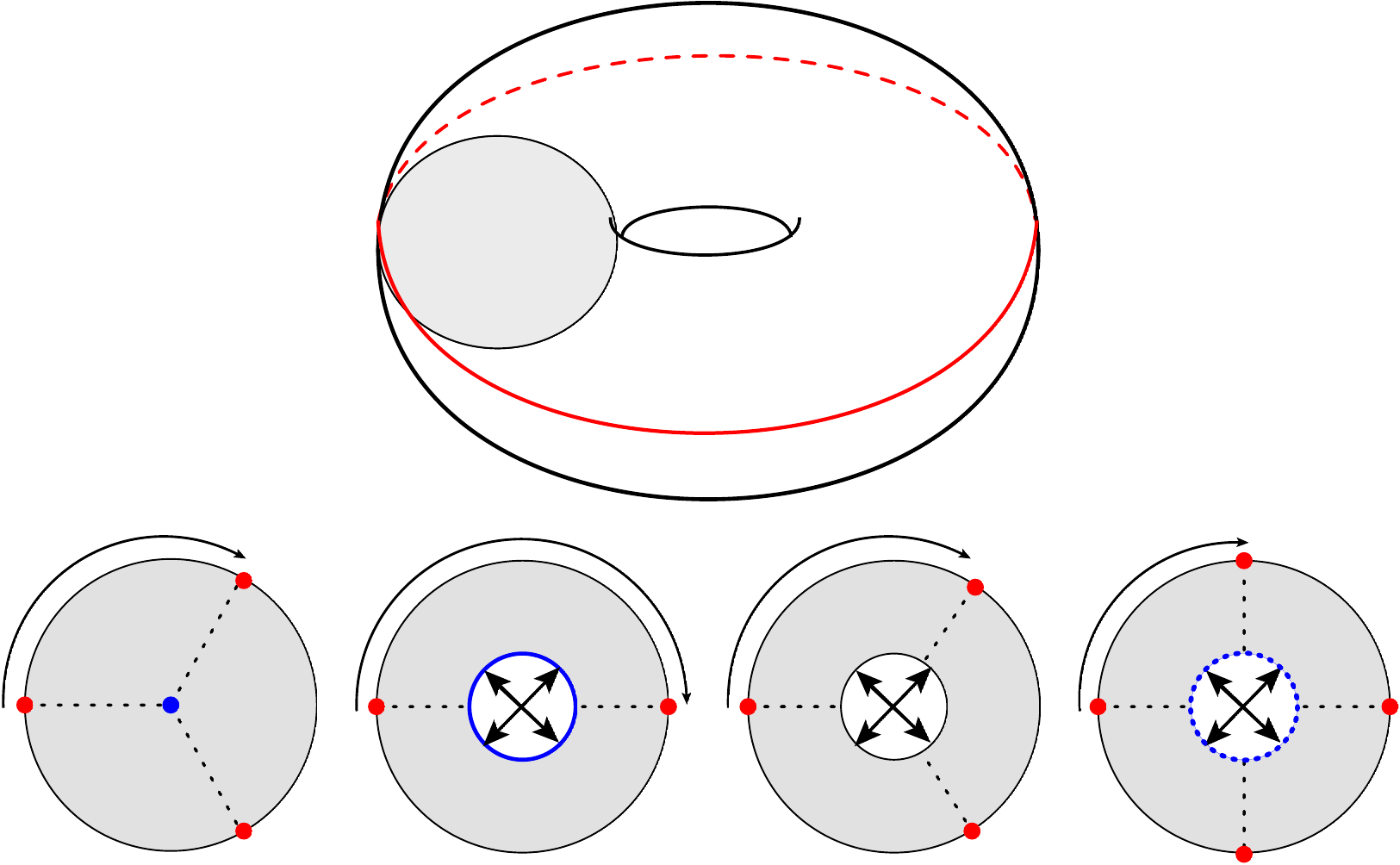}}
\begin{picture}(0.3,0.4)(0,0)
\put(4.85,3.4){\minibox{{\color{red}${\regA}$}}}
\put(1.7,1.4){\minibox{{\color{blue} $\fixM_3$}}}
\put(3.73,1.9){\minibox{{\color{blue} $\fixM_2$}}}
\put(.7,2.4){\minibox{$\beta$}}
\put(3.75,2.65){\minibox{$\beta$}}
\put(5.4,2.4){\minibox{$\beta$}}
\put(7.5,2.2){\minibox{$\beta$}}
\put(1.4,0){\minibox{(a)}}
\put(3.7,0){\minibox{(b)}}
\put(6.1,0){\minibox{(c)}}
\put(8.4,0){\minibox{(d)}}
\end{picture}
\caption{Different ways of filling the boundary torus. The replica $\mathbb{Z}_q$ symmetry acts as a rotation along the Euclidean time direction by $\beta$. Case (a) shows a slice of $\tilde{\bulk}_3$ after filling the boundary time circle with a disk; the replica fixed point set $\fixM_3$ (blue) is a codimension-2 surface in the centre of the circle. Cases (b), (c) and (d) show slices of $\tilde{\bulk}_q$ for $q=2,3,4$, respectively, after filling the boundary torus with a cross cap. For $q=2$ the fixed point set $\fixM_2$ is a codimension-1 orbifold plane wrapping the cross cap. For $q\geq 3$ there is no fixed point set under $\mathbb{Z}_q$ (i.e.\ $\fixM_q$ is empty) and the homology constraint is violated. However, we indicate in (d) that for even $q$ the cross cap itself is still a fixed point set under the subgroup $\mathbb{Z}_2$ of rotations by $\frac{q}{2}\beta$, resulting in an orbifold plane under the quotient.}
\label{fig:crosscap}
\end{figure}
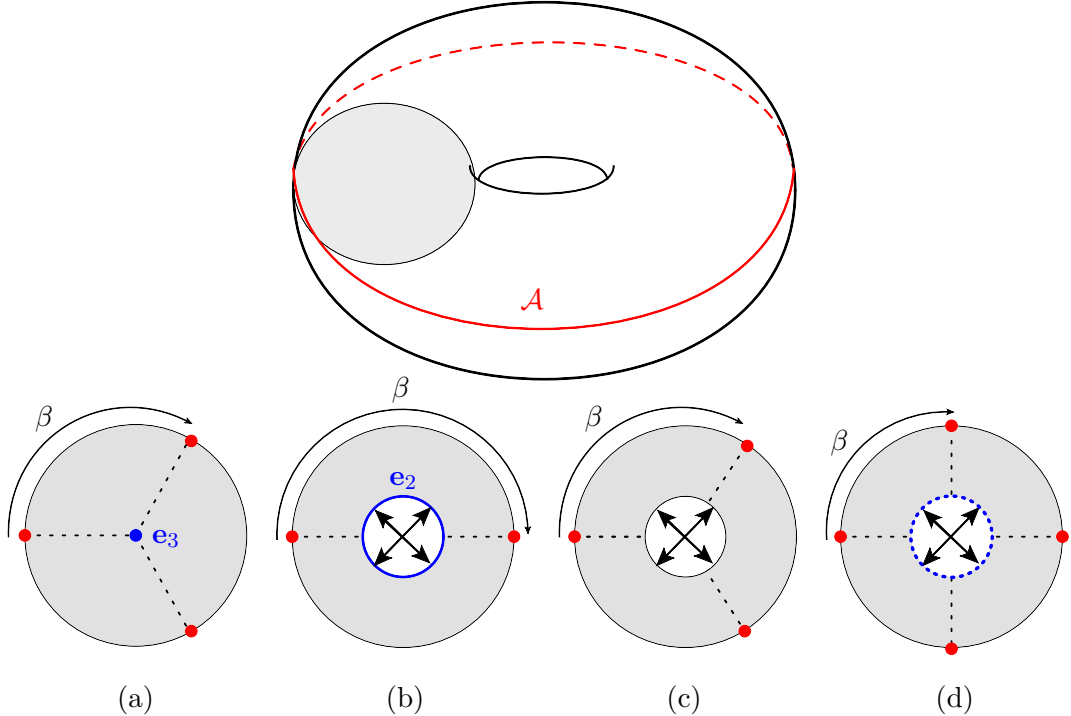

There are now two cases, depending on whether $q$ is odd or even, both of which are illustrated in Fig.~\ref{fig:crosscap}. We begin with $q$ odd. In this case, there are no fixed points of the replica symmetry, and the quotient $\bulk_q$ is a smooth geometry with the same topology as $\tilde{\bulk}_q$. But $\regA$ has non-trivial homology in $\bulk_q$: there is no surface in the bulk whose only boundary is a spatial circle. In this way, the homology constraint is violated because the fixed point set of the replica symmetry (being empty here) is not homologous to the boundary region.

The second case occurs for even $q$.  While no surface is fixed by every element of the replica group, the $q/2$ replica symmetry now rotates half way round the time circle.  On the inner edge of the annulus where points are identified with their antipodes, this symmetry thus takes every point to itself.  In the quotient $\bulk_q$, the resulting singular set is the inside circle of the annulus times the spatial circle.  This set is not codimension-2, but instead codimension-1; it is a $\mathbb{Z}_2$ orbifold plane. Thus, for even $q$ there at least exists a fixed point set under a $\mathbb{Z}_2$ subgroup of $\mathbb{Z}_q$, despite the fixed points of the full replica group $\mathbb{Z}_q$ being empty for $q>2$.

There are several objections that one might raise to this example. The most obvious is that a metric is never constructed, and that there can never be any genuine saddle points with the given topology. The clearest refutation of this is to give an example sharing all the same qualitative features, but with a metric. One can in fact  give such a construction for pure three dimensional gravity. To explain the idea, consider the example given instead by starting from a geometry with two disjoint torus asymptotic boundaries, and then taking a $\mathbb{Z}_2$ quotient by swapping the two tori and simultaneously rotating half way round the time circle. This gives a cross-cap in the time circle as described here. Now generalize it to start with not tori, but higher genus Riemann surfaces with negative Euler characteristic.  There is an easy way to put a constant negative curvature metric on these geometries, as in the Maldacena-Maoz wormhole \cite{Maldacena:2004rf}. If the Riemann surface has a fixed-point free involutive isometry, the $\mathbb{Z}_2$ quotient of this combined with swapping the two boundaries gives the relevant example.

A second objection is that this fails to satisfy the homology constraint in a very particular way, by having no fixed points under the replica symmetry. This is impossible if the region $\regA$ has a nonempty boundary, since there must be a set of fixed points extending from $\entsurf$ into the bulk. We will show later in this section that it is possible to circumvent this objection and construct an example with $\extr$ nonempty, but not homologous to $\regA$.

Finally,  these geometries might never be dominant saddles in the path integral. While this is plausible, it is very difficult to see how their dominance can be ruled out in general. For this reason, we would like to make arguments that apply using only the topology.

\subsection{Implications for the homology constraint}

We believe this example has a real lesson to teach us: the homology constraint follows from the holographic replica trick only when the defect arises from a $\mathbb{Z}_q$ quotient \emph{at every positive integer $q$}. In the context of the LM argument, one may start with a bulk geometry $\bulk$ with boundary $\bdy$, and pick some extremal surface $\extr$ ending on $\entsurf$. Now introduce a small conical defect along $\extr$, and adjust the geometry so it remains on-shell away from the defect, which we now call $\fixM_q$. The angle of the conical singularity can be dialed to $\frac{2\pi}{q}$ for real positive values of $q$, to give a family of singular bulks $\bulk_q$, returning to $\bulk$ when $q\to1$. We assume that this can be done without changing the topology of the bulk, or the defect within it. When $q$ hits an integer, we can make a connection with the replica trick, but only if $\bulk_q$ can be lifted to a branched cover. That is, we require that there exists some $\tilde{\bulk}_q$, with boundary $\tilde{\bdy}_q$, and with a $\mathbb{Z}_q$ symmetry such that the quotient yields $\bulk_q$. Under this condition the action evaluated on $\tilde{\bulk}_q$ will legitimately give the $q^{\rm th}$ R\'enyi entropy (assuming $\tilde{\bulk}_q$ is the dominant saddle).

The examples tell us that this may happen at some $q$ despite $\extr$ violating the homology constraint. But taking $\bulk$ as the cross-cap geometry above, with $\extr$ chosen to be empty, this lift to $\tilde{\bulk}_q$ can be constructed only at odd $q$, and does not exist for even $q$. This is because in the lift we must choose which copy of the boundary to end up on after passing through the cross-cap, say by rotating through $k$ copies. But going through the cross-cap twice is homotopic to passing round the boundary circle, so must take us through one copy. This means we require $1= 2k \mod q$, which has a solution if and only if $q$ is odd, so no choice of $k$ gives the correct cover on the boundary. We argue below that demanding the stronger condition  that the lift must exist for all $q$ so that $\bulk_q$ can be used to compute $S^{(q)}_\regA$ at every positive integer, is equivalent to the homology constraint on $\extr$.

This example generalizes to allow nonempty $\extr$ as follows. Take the bulk $\bulk$ to be the same geometry as above, a solid torus with a neighborhood of a loop round the torus cut out and replaced with a M\"obius strip times a circle. Consider now the region $\regA$ as not a whole boundary circle, but an interval. Choose the surface $\extr$ to be a curve joining the endpoints of $\regA$, but passing round the non-contractible spatial circle on the side of the torus opposite of $\regA$, and hence not homologous to it. Nonetheless, again at odd values of $q$ this lifts to a smooth replica-symmetric bulk $\tilde{\bulk}_q$ in much the same way. When going round the surface $\extr$, move from sheet to sheet as usual, but when passing through the crosscap pass to the $(\frac{q+1}{2})^\text{th}$ sheet relative to where you begin. In essence, the crosscap acts to replace the closed geodesic that would otherwise be required, but without the associated fixed point set or conical defect.

\subsection{Further examples}

There are several aspects of the above examples that may appear to be required to get the sort of homology violation we see: for example, a non-orientable spacetime, a distinction only between odd and even $q$, or the specific way a topological feature replaces a part of the spacetime where a fixed point `should be'. We now briefly describe an example to show that none of these are necessary. Take $\bulk$ as a 3-dimensional ball, with boundary $\bdy=S^2$, cut out a ball at the center, and then identify antipodal points of the spherical edge of the resulting hole. This results in an $\mathbb{RP}^3$ with a boundary, an orientable spacetime. Take $\regA$ as an interval on the boundary sphere, running between the poles perhaps, and choose $\extr$ to pass through the nontrivial topology we have introduced in the centre. This can be lifted to a replica-symmetric cover\footnote{An easy way to see this is to notice that $\bulk-\extr$ deformation retracts to a circle, and going twice round this circle (or $p$ times in the generalization) is homotopic to a boundary loop passing through $\regA$ once.} to compute R\'enyi entropies for odd $q$ as above, but not for even $q$. A generalization identifies the points on the sphere bounding the cut out region as in the construction of a Lens space $L(p;p')$ from identifications of the boundary of a 3-ball (the special case of antipodal points being $p=2$). The result of this is that a lift will exist for $q$ co-prime to $p$ and not otherwise.

We conclude this section with a different sort of example, which shows that it is possible to obtain a codimension-2 defect with deficit angle different from $\frac{2\pi}{q}$ from the quotient construction. The simplest example generalizes the computation of thermal entropy from BTZ; again we will not explicitly mention the trivial dependence on the spatial direction and focus on a constant spatial slice. The usual geometry then picks $\extr$ as a point in the centre of this slice, and passing anticlockwise round this point takes you up one sheet in the cover. To generalize this, choose $\extr$ as a collection of several points, and for each point choose some integer number of sheets to change by when passing round them anticlockwise. To be consistent with the boundary covering space, the only requirement is that these integers sum to one. But if the integer chosen for a point, say $n$ is anything other than $\pm 1$, the quotient of the covers thus constructed will result in a variety of conical deficit angles, depending on the greatest common factor $r$ of $n$ and $q$. Specifically, the covering space will contain $r$ copies of the point, related by the ${\mathbb Z}_r$ subgroup of the ${\mathbb Z}_q$ replica symmetry, and the resulting defect angle after the quotient will be $\frac{2\pi r}{q}$.

For a specific example, take two points, choosing to increase by two sheets when passing round the first, and to decrease by one when passing round the second. The $q=4$ covering space has the topology of a torus with a single boundary (being the usual four glued copies of the boundary circle). This can be understood as a regular octagon, with top and bottom edges identified, as well as left and right edges; the remaining diagonal edges form the boundary circle. The replica symmetry acts by rotations by
$\frac{\pi}{2}$. This leaves the obvious fixed point at the centre, which becomes the usual $\frac{\pi}{2}$ deficit after the quotient. The replica symmetry rotating by $\pi$ has an additional two fixed points, being the centre of the top or bottom edge (identified with one another), and the centre of the left or right edge. Under the quotient, these two points become identified, and the result is a conical defect of angle $\pi$.

These sorts of examples can be thought of as many degenerate extremal surfaces lying on top of one another, which is very physically natural when viewing the defects as cosmic strings, and can be recovered as a limit of several of the usual defects coalescing. We would not expect these configurations to dominate in simple examples, particularly in the $q\to 1$ limit, but we can not rule out the possibility that these surfaces with multiplicity may be favourable enough to dominate in some more complicated geometry.

\section{A topological condition on $q$-R\'enyi saddles}
\label{sec:consistency}

In light of the example described in \S\ref{sec:crosscap}, we now wish to formulate a refined requirement regarding existence of covering spaces consistent with replica symmetry, which will be sufficiently strong to impose the expected topological constraint on the branching surface.

\subsection{Boundary conditions for branched covers}

The replica trick on the boundary involves defining the CFT on a branched cover $\tilde{\bdy}_q$.  To compute the semiclassical R\'enyi entropy, one should evaluate the gravitational action on every bulk manifold $\tilde{\bulk}_q$ satisfying the equations of motion and the boundary condition $\partial \tilde{\bulk}_q = \tilde{\bdy}_q$, and choose the dominant saddle. This is prohibitively difficult, so we ignore the question of which saddle dominates, and furthermore we restrict to $\tilde{\bulk}_q$ which can be realized as a branched cover over the original space $\bulk$. That is, we assume that $\bulk_q = \tilde{\bulk}_q / \mathbb{Z}_q$ is homeomorphic to the original manifold $\bulk$.
The boundary condition $\partial \tilde{\bulk}_q = \tilde{\bdy}_q$ imposes a topological condition on
$\cutout{\bulk_q}{\fixM_q}$, which we have seen above is \emph{not} the homology constraint.  First we need to understand this condition on $\cutout{\bulk_q}{\fixM_q}$ more precisely.

Let $\bdy$ be the boundary, and $\regA$ the region whose entanglement entropy we want to compute. Consider arbitrary closed loops in $\bdyc$, i.e., loops which may intersect $\regA$, but not $\partial \regA$. Then there is a homomorphism $\phi:\pi_1(\bdyc)\to \mathbb{Z}$ which counts the number of times the loop passes through $\regA$. This map is defined to take into account the orientation of the loop relative to $\regA$, so it computes the signed intersection number. Hence for any given $q$ there is a map $\phi_q:\pi_1(\bdyc)\to \mathbb{Z}_q$ which is just the previous map modulo $q$. The replicated boundary can be found by taking the cover $\bdycCover$ of $\bdyc$, defined so that loops in $\ker\phi_q$ lift to closed loops in the covering space. This is just saying that $\phi_q$ counts which sheet we are on in the replicated CFT, and loops which intersect $q$ times with $\regA$ are to be identified with closed
loops in the $q$-fold covering space in accordance with the replica trick.

Replica symmetry and the boundary condition $\partial \tilde{\bulk}_q = \tilde{\bdy}_q$ imply a similar construction for the bulk. First, it implies that there exists a map $\psi_q$ keeping track of movements between sheets of the cover in the bulk. Second, it requires that this restricts to $\phi_q$ on the boundary $\partial (\bulkcfix{q}) = \bdyc$ to give the correct boundary conditions.
This can be summarized as the existence of a bulk sheet-counting homomorphism $\psi_q$, factoring $\phi_q$ so that the following diagram commutes:
\begin{equation} \label{eq:ConsCondAA}
\begin{tikzcd}
\pi_1(\bdyc) \arrow{r}{\phi_q}\arrow{d}{i_*} & \mathbb{Z}_q \\
\pi_1(\bulkcfix{q}) \arrow{ur}{\psi_q}
\end{tikzcd}
\end{equation}
Here $i_*$ is the pushforward induced by the inclusion $i$ of the boundary into the bulk. Below, we show more formally that the existence of such a $\psi_q$ is equivalent to the existence of a replica-symmetric covering space of the bulk.

\subsection{Stronger criterion from the cosmic brane construction}

The condition \eqref{eq:ConsCondAA} is true of any branched cover satisfying the boundary conditions required by AdS/CFT.  The example in \S\ref{sec:crosscap} illustrates that with this condition alone, the singularity $\fixM_q$ is not necessarily homologous to $\regA$.

But now, let us we restrict to branched covers of the original manifold $\bulk$ constructed following LM: choose a codimension-2 surface $\extr$, and introduce a conical defect $\frac{2\pi}{q}$ at this surface.  Can this be lifted to a branched cover $\tilde{\bulk}_q$ obeying the correct boundary conditions? In general, this is possible only if the boundary condition \eqref{eq:ConsCondAA} holds for $\cutout{\bulk}{\extr}$. This criterion must be applied at each $q$ separately, as in general it may be possible to find $\psi_q$ for some values of $q$ but not for others, as in the crosscap example in \S\ref{sec:crosscap}. If $\psi_q$ exists lifting $\phi_q$ for all $q$, then there is a lift\footnote{ This is guaranteed if the relevant groups are finitely generated, which holds for compact manifolds with boundary.} $\psi:\pi_1(\cutout{\bulk}{\extr})\to \mathbb{Z}$ of $\phi$, meaning that the following diagram commutes:
\begin{equation}\label{eq:ConsCond}
\begin{tikzcd}
\pi_1(\bdyc) \arrow{r}{\phi}\arrow{d}{i_*} & \mathbb{Z} \\
\pi_1(\cutout{\bulk}{\extr}) \arrow{ur}{\psi}
\end{tikzcd}
\end{equation}
We will prove that \eqref{eq:ConsCond} is equivalent to the homology condition for the RT surface.

To illustrate this, let us return to the crosscap example of \S\ref{sec:crosscap} in this language. In this case, the starting point $\bulk$ is the M\"obius strip times a circle. We choose the defect $\extr$ to be the empty set, and try to construct branched covers at integer $q$. This can be done in a way satisfying the boundary condition \eqref{eq:ConsCondAA} for odd $q$, as illustrated in Fig.~\ref{fig:crosscap}, but for even $q$, there is no way to construct the branched cover. This can be seen very explicitly: we have $\phi_q(\Gamma)=1$, where $\Gamma$ is the loop going once around the boundary time circle. Assume there exists a lift $\psi_q$ and require $\psi_q(\Gamma)=1$. Now try to find a consistent way of choosing a value under $\psi_q$ for the bulk loop $\gamma$ that goes once through the cross-cap: since $\Gamma$ is homotopic to  $2\gamma$, we require $1 \,(\text{mod } q)= \psi_q(\Gamma) = \psi_q(2\gamma) = 2\,\psi_q(\gamma)$ which has a solution in $\mathbb{Z}_{q}$ if and only if $q$ is odd. At the level of the maps $\phi$ and $\psi$ of \eqref{eq:ConsCond}, this contradiction for even values of $q$ manifests itself as follows: both fundamental groups in \eqref{eq:ConsCond} are $\mathbb{Z}$ (ignoring the spatial circle), but one of them is generated by $\Gamma$, the other is generated by $\gamma$. This means $\phi$ is the identity (in particular $\phi(\Gamma)=1$) but $i_*$ is multiplication by 2, so clearly $\psi$ does not exist.

\subsection{Another look at the boundary condition}
\label{sec:detailsCovers}

Above we motivated \eqref{eq:ConsCondAA} from the AdS/CFT boundary conditions.  We now discuss more formally how to show that, for a codimension-2 defect, this condition is equivalent to the existence of a topologically replica-symmetric covering space of the bulk.

We take here the perspective that a branched cover is defined as a covering space in the usual sense, where any point in the base space has a neighborhood lifting to $q$ homeomorphic copies of itself, but after removing the branching surface.\footnote{ This works when the branching surface is codimension-2 as expected, since then there is a uniquely specified way to put back the branching surface in the cover. It fails when the branching surface is codimension-1, since extra information on what happens when passing through the surface must be imposed, but as this can only happen for even $q$ in any case, this will not alter our main conclusions. Similarly, there are examples with fixed point sets of codimension greater than 2, but the local topology of such a quotient depends on $q$, so that there is no obvious way to fit them into a family at all values of $q$.} Given this perspective, a replica-symmetric $q$-fold cover of $\bulk_q$, branched at $\fixM_q$, restricting correctly on the boundary, is defined by the following diagrams:
\begin{equation}\label{eq:ConsCondBB}
\begin{tikzcd}
\bdycCover \arrow{r}{p} \arrow{d}{\tilde{i}}
&\bdyc \arrow{d}{i}\\
 \bulkcCover{q} \arrow{r}{P} & \bulkcfix{q}
\end{tikzcd}\qquad
\begin{tikzcd}
\pi_1(\bdycCover) \arrow{r}{p_*} \arrow{d}{\tilde{i}_*}
&\pi_1(\bdyc) \arrow{r}{\phi_q} \arrow{d}{i_*}
&\mathbb{Z}_q \arrow{r}
&1\\
\pi_1(\bulkcCover{q}) \arrow{r}{P_*} &
\pi_1(\bulkcfix{q}) \arrow{r}{\psi_q}
&\mathbb{Z}_q \arrow{r}& 1
\end{tikzcd}
\end{equation}

The left hand diagram gives the covering space for boundary (top row) and bulk (bottom row), with respective covering maps $p$ and $P$. Then $i$ and $\tilde{i}$ are the injective inclusion maps of boundary into bulk, and the diagram commutes to implement the inclusions consistently. This implies immediately existence and commutativity of the left square of the diagram on the right. The additional maps are required to make the rows exact; the existence of $\phi_q$ and $\psi_q$ is then equivalent to replica symmetry. This is because we require the $q$-fold cover to possess the replica symmetry, acting via a $\mathbb{Z}_q$ deck transformation group, acting cyclically. This implies that the covering space is \emph{normal}, so that $P_*(\pi_1(\bulkcCover{q}))$ is a normal subgroup of $\pi_1(\bulkcfix{q})$, and taking the quotient by that subgroup gives the $\mathbb{Z}_q$ deck transformation group. This shows that the map $\psi$ exists, being the quotient map. The same argument also applies on the boundary. In particular $\ker(\phi_q) = \im(p_*)$ and $\ker(\psi_q) = \im(P_*)$. This shows that exactness is just a formal restatement of replica symmetry (in a topological sense) both at the boundary and in the bulk.

We prove in Appendix \ref{sec:theoremsproofs} that the holographically natural consistency condition \eqref{eq:ConsCondAA} holds if and only if $\bulkcfix{q}$ is such that the diagrams \eqref{eq:ConsCondBB} exist as described: existence of the map $\psi_q$ to factor $\phi_q$ as $\psi_q \circ i_*$ implies that a cover can be constructed, and conversely existence of a cover implies that there is such a $\psi_q$.

\section{Relation between topological consistency and homology constraint}
\label{sec:homology}

Having dealt with the topological consistency condition required for the LM construction, we now turn to the central thesis of this work: ``when is the homology condition satisfied?'' We will give an overview of our general strategy first and exemplify it with the BTZ spacetime. In \S\ref{sec:homology_derivation} we prove Theorem \ref{thm:homology} which posits that we are guaranteed the homology constraint provided the topological consistency condition is met (and vice-versa).

\subsection{General strategy}
\label{sec:HomologyOutline}
The extension of the replica trick into the bulk \cite{Lewkowycz:2013nqa} constructs a bulk geometry with boundary conditions given by the replicated field theory. At integer $q$ this leads to a geometry which is a smooth $q$-fold covering space $\tilde{\bulk}_q$ of the original bulk geometry. The quotient $\bulk_q = \tilde{\bulk}_q/\mathbb{Z}_q$ is then just the original bulk geometry with a conical defect inserted along the codimension-2 fixed point set of $\mathbb{Z}_q$.
 By construction $\bulk_q$ lifts to a $q$-fold branched cover of $\bulk$ at every integer value of $q$ and the topological consistency condition \eqref{eq:ConsCond} holds. This topological consistency condition is the essential feature of the LM construction on which we will focus for the remainder of this section. In particular, we will no longer concern ourselves with covering spaces and simply restrict attention to topological properties of $\bulk$ and $\extr$.

For concreteness, consider the following setup. We start with a bulk geometry $\bulk$, a boundary region $\regA$ and a candidate extremal surface $\extr$ to be used in computing entanglement entropy of $\regA$. The goal of the present section is to first illustrate and then prove the following statement: \emph{The topological consistency condition is satisfied for all $q$ if and only if $\regA$ is homologous (in the sense of \eqref{eq:radef}) to the surface $\extr$.}\footnote{ We postpone subtleties concerning orientation to \S\ref{sec:NonOrientable}.}
In order to decide whether $\bulk-\extr$ satisfies the topological consistency condition for all $q$ (and thus whether inserting a conical defect $\frac{2\pi}{q}$ along $\extr$ would lift to a $q$-fold branched cover in the sense of LM at all integer $q$), consider the following maps:
 \begin{itemize}
 \item A \emph{boundary sheet counting map} $\phi \in H^1(\bdy-\partial\regA)$: this is the map from \eqref{eq:ConsCond} which counts how many times a boundary loop passes through $\regA$.  Although $\phi$ was originally defined on the homotopy group $\pi_1(\bdy-\partial\regA)$, the fact that it maps to an abelian group ($\mathbb{Z}$) makes its action on the first homology group well-defined.  It is an element of the first cohomology group $H^1(\bdy-\partial \regA)$ of $\bdy-\partial\regA$ with integer coefficients; such elements are just homomorphisms from boundary loops into $\mathbb{Z}$.
 \item A \emph{local intersection map} $\uE \in H^2(\bulk,\bulk-\extr)$.    By $H^2(\bulk,\bulk-\extr)$ we mean the second cohomology group of $\bulk$ relative to $\bulk-\extr$ with integer coefficients. This $\uE$ is defined on 2-dimensional surfaces (2-chains) in a neighborhood of $\extr$ whose boundary is not a part of $\extr$. Given such a 2-surface $\disk$, the map $\uE$ counts the (signed) number of intersections of $\extr$ with $\disk$.\footnote{ Note that this is only true in the absence of torsion cycles. If the bulk spacetime has torsion (in the topological sense), then the second cohomology group is not isomorphic to intersection counting homomorphisms. To illustrate the essence of our argument, we refrain from considering these subtleties at the present. We will however account for torsion in our proof of Theorem \ref{thm:homology} in \S\ref{sec:homology_derivation}.} The sign of an intersection is given by the relative orientation of $\disk$ and $\extr$.

 \end{itemize}
 Intuitively, $\phi$ (being defined by $\regA$) encodes the boundary conditions for  bulk geometries that can be lifted to replica symmetric covers. On the other hand, $\uE$ (being defined by $\extr$) carries topological information about the way such a bulk covering space would be branched. The topological consistency condition relates these two objects: namely, we can translate the topological consistency condition into a certain consistency of $\phi$ with respect to $\uE$. To this end, take any 2-surface $\disk$ in the bulk, which is anchored outside of $\partial \regA$ on the boundary, i.e., $\partial \disk \subset \bdy-\partial\regA$. Now compare the following two properties of $\disk$:
 \begin{itemize}
 \item Use $\phi$ to count how many times the loop $\partial\disk$ passes through $\regA$ on the boundary (taking into account orientations).
 \item On the other hand, consider the parts of $\disk$ which lie in a tubular neighborhood of $\extr$ and compute their (signed) intersection number with $\extr$ using $\uE$.
 \end{itemize}
 Our main statements can be summarized as follows:
 \begin{enumerate}
 \item The topological consistency condition means that these two countings have to agree for any 2-surface $\disk$ anchored at the boundary.
 \item The two countings agree for all 2-surfaces $\disk$ anchored at the boundary if and only if $\regA$ and $\extr$ are homologous.
 \end{enumerate}
 Illustration and proof of these points will be our main task.
 The rationale behind the first point, is roughly the following; it is a reformulation of the consistency condition at the level of the unreplicated bulk geometry. In fact, the topological consistency condition as formulated in \eqref{eq:ConsCond} means exactly that $\phi$ can be lifted to a global sheet counting map $\psi \in H^1(\bulk-\extr)$ for bulk loops in a way that is consistent with $\uE$ on loops that bound 2-surfaces which intersect $\extr$. Roughly speaking, 2-surfaces intersecting $\extr$ encode the same information about local branching near the intersection point as do small loops going around the intersection point. We thus require that $\uE$ is inherited from the global sheet counting map $\psi$; technically we need $\uE = \delta\psi$ where $\delta$ is the coboundary map as defined in Appendix \ref{sec:algebraictopology}. This idea will be made precise in \S\ref{sec:homology_derivation}.

 In the rest of this section we proceed as follows. We first illustrate the above ideas using the example of the BTZ black hole. Here we refrain from rigor and just give qualitative arguments for the validity of the argument. \S\ref{sec:homology_derivation} then presents a rigorous and general proof that the topological consistency condition is satisfied if and only if $\regA$ is homologous to $\extr$ in an appropriate sense.

\subsection{Example: BTZ black hole}
\label{sec:BTZ}

Let us consider the single interval entanglement entropy of a CFT${}_2$ on a spatial circle and at finite temperature, in the deconfined phase above the Hawking-Page transition. The Euclidean bulk $\bulk$ is topologically a solid torus. For illustration, consider the various extremal surfaces sketched in Fig.~\ref{fig:btz}. In each case we want to study the relation between our topological consistency condition (\ref{eq:ConsCond}) and the homological properties of $\extr$ compared to $\regA$. The geodesics $\extr^{(a)}$ and $\extr^{(c)}$ are both homologous to $\regA$ (i.e., in each case there exists an interpolating surface whose only boundaries are $\regA$ and $\extr^{(a,c)}$) and should therefore be allowed saddles in the RT formula. As we will see, both of them satisfy the topological consistency condition for all $q$ and are thus consistent with the LM construction. The extremal surface $\extr^{(b)}$, on the other hand, is not homologous to $\regA$ which will manifest itself as a breakdown of the topological consistency condition (meaning that $\extr^{(b)}$ can never arise as the fixed point set of a LM-type argument).

\begin{figure}
\setlength{\unitlength}{0.1\columnwidth}
\centerline{\includegraphics[width=\textwidth]{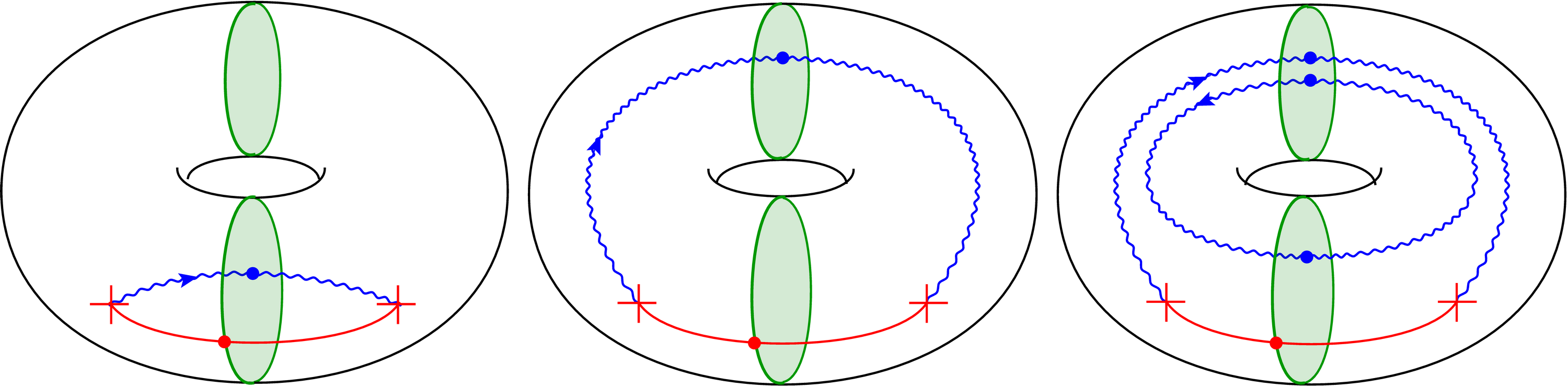}}
\begin{picture}(0.3,0.4)(0,0)
\put(1.9,.8){\minibox{{\color{red}${\regA}$}}}
\put(.9,1.3){\minibox{{\color{blue}${\extr^{(a)}}$}}}
\put(3.9,1.8){\minibox{{\color{blue}${\extr^{(b)}}$}}}
\put(7.65,2){\minibox{{\color{blue}${\extr^{(c)}}$}}}
\put(1.53,2.2){\minibox{{\color{darkgreen}$\disk$}}}
\put(1.53,.8){\minibox{{\color{darkgreen}$\disk$}}}
\put(1.45,.1){\minibox{(a)}}
\put(4.85,.1){\minibox{(b)}}
\put(8.2,.1){\minibox{(c)}}
\end{picture}
\caption{Three different candidates for bulk geodesics (blue) whose length may be considered to compute the entanglement entropy of the spatial boundary interval ${\regA}$ (red). The green surfaces $\disk$ illustrate the connection between boundary sheet counting (intersections with $\regA$) and bulk intersections. Thick dots illustrate intersections. The bulk surfaces $\extr^{(a)}$ and $\extr^{(c)}$ are homologous to $\regA$; which one of them is dominant depends on the length of $\regA$. The surface $\extr^{(b)}$ is not homologous to ${\regA}$ and is also forbidden by our topological consistency condition. In this latter case, boundary intersections do not match bulk intersection numbers.}
\label{fig:btz}
\end{figure}
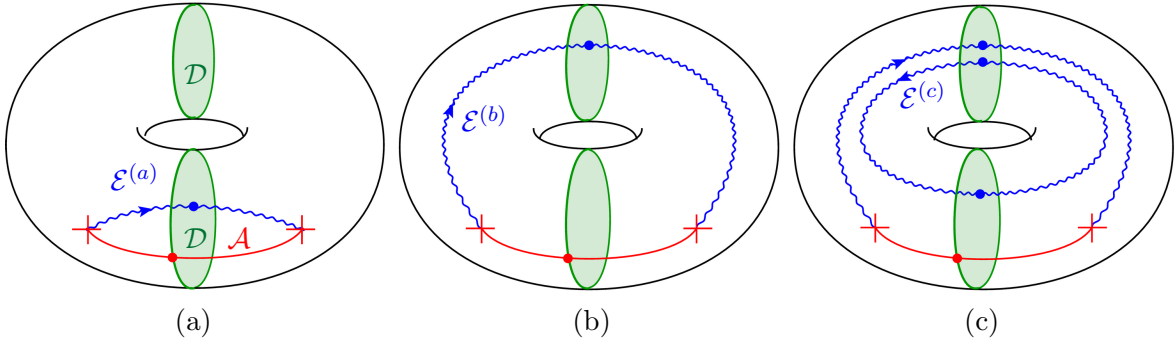

By inspection, it is evident from the examples in Fig.\ \ref{fig:btz} how homology of $\regA$ and $\extr$ are tied to consistency of the boundary sheet counting map $\phi$ and the bulk intersection map $\uE$. For instance, in case (a) we have drawn two choices of surfaces $\disk$ with boundary on $\bdy-\partial\regA$ and clearly for both of them the number of intersections of $\partial\disk$ with $\regA$ matches the local intersection number of $\disk$ with $\extr^{(a)}$ in the bulk. The reader can easily verify that this matching indeed holds for any choice of 2-surface $\disk$ anchored on $\bdy-\partial\regA$. This consistency of $\phi$ and $\uE$ allows for a global sheet counting map $\psi \in H^1(\bulk-\extr)$ which reduces to $\phi$ at the boundary.

Similar reasoning holds for case (c): The only novelty here is the fact that orientations need to be taken into account to properly count intersections. For instance, the intersection number of the rear disk with $\extr^{(c)}$ is zero due to opposite orientations of $\extr^{(c)}$ relative to $\disk$ at the two intersection points. This matches the fact that the boundary of the rear disk does not intersect with $\regA$.

Now consider case (b) where $\extr^{(b)}$ indeed violates the homology constraint. As we can see, there now exist 2-surfaces $\disk$ on which $\phi$ (intersections of $\partial\disk$ with $\regA$) takes a different value than $\uE$ (intersection number of $\disk$ with $\extr^{(b)}$). Therefore the boundary sheet counting map $\phi$ does not lift to a global sheet counting map $\psi$ consistent with $\uE$ on loops which are boundaries of 2-surfaces that intersect $\extr^{(b)}$. Such a lift not being possible means that covers of the bulk branched along $\extr^{(b)}$ do not have the correct boundary conditions. This situation can therefore never arise from a LM construction.

At a pictorial level, this demonstrates our claims for the BTZ geometry: the homology constraint on $\regA$ and $\extr$ comes in conjunction with consistency of bulk intersection number with $\extr$ on the one hand, and boundary intersection number with $\regA$ on the other hand. This consistency is equivalent to the topological consistency condition \eqref{eq:ConsCond}. We will now turn to a rigorous proof of this idea.

\subsection{Topological consistency is equivalent to homology constraint}
\label{sec:homology_derivation}

We begin with a reminder of the notation. The bulk $\bulk$ is a $d+1$ dimensional orientable manifold with boundary $\bdy$, $\regA$ is a $d-1$ dimensional submanifold of $\bdy$ with boundary $\entsurf$, and $\extr$ is a $d-1$ dimensional submanifold of $\bulk$, also with boundary $\entsurf$. It should be borne in mind that either or both of $\regA$ and $\extr$ may be disconnected. For this proof, we will focus on the case when the bulk is orientable; we will return to the extension to the non-orientable case in the next subsection. A review of the algebraic topology required for this section can be found in Appendix \ref{sec:algebraictopology}.

To construct the cover on the boundary, we have a map from $\pi_1(\bdy-\entsurf)$ to $\mathbb{Z}$, counting signed intersections with $\regA$. This is equivalent to a map $H_1(\bdy-\entsurf)\to\mathbb{Z}$, since this homology group\footnote{ All homology and cohomology groups are taken with coefficients in $\mathbb{Z}$, so we will not explicitly indicate this dependence until we later generalize to the non-orientable case.} is the abelianization of the fundamental group by the Hurewicz theorem. In turn, by the universal coefficient theorem, this homomorphism is equivalent to a cohomology class $\phi\in H^1(\bdy-\entsurf)$.

This sheet counting map is equivalent to a homological description of the region $\regA$. Regard $\bdy-\entsurf$ as the boundary minus a tubular open neighbourhood of $\entsurf$, so it becomes a compact manifold with boundary. We may then use Poincar\'e-Lefschetz duality \eqref{eq:PoincareLefschetz} to send $\phi\in H^1(\bdy-\entsurf)$ to $\reghom\in H_{d-1}(\bdy-\entsurf,\partial(\bdy-\entsurf))\approx H_{d-1}(\bdy,\entsurf)$. This is a homology class represented by an appropriately chosen orientation of $\regA$.

The arguments above show that the bulk lifts to an appropriate branched cover at all $q$ whenever there exists some extension of $\phi$ into the bulk. So we require some $\psi\in H^1(\bulk-\extr)$ satisfying $i^*\psi=\phi$, where $i:\bdy-\entsurf\to\bulk-\extr$ is the inclusion of boundary into bulk, so that the restriction of $\psi$ to the boundary is $\phi$.

In addition to this, we need to know what happens when we traverse a small loop passing round the branching surface $\extr$. We make this precise by considering the normal bundle of $\extr$, which can be embedded in $\bulk$ as some tubular neighbourhood of $\extr$. The fibres $F$ of the bundle are copies of $\mathbb{R}^2$, with $\extr$ lying at the origin, so the action of $\psi$ passing round the surface is described by its restriction to $F-\{0\}$ on each fibre, $\psi_F \in H^1(F-\{0\})\approx \mathbb{Z}$. We would normally like this restriction to each fibre to be a generator of $H^1(F-\{0\})$, which implies that after taking the $q$-fold cover and the $\mathbb{Z}_q$ quotient there will be a $\frac{2\pi}{q}$ conical deect at every $q$. This need not hold in general, as examplified in \S\ref{sec:crosscap}, and the conical defect could be of a different angle depending on $q$ in a more complicated way. This can be accounted for in the end by counting such a surface with the appropriate multiplicity in the homology computation, and by allowing this multiplicity we may assume that $\psi_F$ is always a generator.

So we would like necessary and sufficient conditions for the existence of some $\psi\in H^1(\bulk-\extr)$, such that
\begin{enumerate}
\item The restriction $i^*\psi$ to the boundary is $\phi\in H^1(\bdy-\entsurf)$, and
\item The restriction $\psi_F$ to the embedding into $\bulk$ of the fibres $F$ of the normal bundle $N$ of $\extr$ is a generator of $H^1(F-\{0\})$.
\end{enumerate}
The choice of $\psi_F$ will in the course of the proof supply us with an appropriate orientation of $\extr$, and thus an element $\extrhom\in H_{d-1}(\extr,\entsurf)$. We now have enough to state our main result.

\begin{theorem}
\label{thm:homology}
There exists some $\psi\in H^1(\bulk-\extr)$ which restricts to $\phi$ on the boundary, and which restricts to $\psi_F$ on the fibres of $N$, if and only if $\regA$ is homologous to $\extr$, in the sense that the inclusions of $\extrhom$ and $\reghom$ into $H_{d-1}(\bulk,\entsurf)$ are equal.
\end{theorem}

\begin{proof}

The choice of generators of $H^1(F-\{0\})$ on each fibre is equivalent to an orientation of each fibre, which can be regarded equivalently as a smooth choice of generators of the relative cohomology groups $H^2(F,F-\{0\})$, via a coboundary map. In each connected component of $\extr$, making a choice on one fibre leaves no freedom on the others, which are determined uniquely from continuity. From this we can get a well-known object in the topology of vector bundles, a Thom class $u\in H^2(N,N-\extr)$, where $\extr$ here sits at the zero section. This follows since if $\extr$ is connected, then the restrictions $H^2(N,N-\extr)\to H^2(F,F-\{0\})$ are in fact all isomorphisms. This is clear when restricted to a local trivialization of the bundle, and can be extended to the whole of the connected component by a standard procedure of gluing trivializations together one by one using a Mayer-Vietoris sequence. So our choices of $\psi_F$ can be summarized by a choice of generator $u\in H^2(N,N-\extr)$ for each connected component of $\extr$. See section 4.D of \cite{Hatcher:2002} for a discussion of the Thom class.

Now regarding $N$ once again as a tubular neighbourhood\footnote{ The existence of such a neighborhood is in fact a technical assumption, though this can be relaxed \cite{samelson1965}.} of $\extr$ in $\bulk$, we can use an excision theorem to give an isomorphism $H^2(N,N-\extr)\approx H^2(\bulk,\bulk-\extr)$ (excising $\bulk-\extr$), so we may regard $u$ instead as an element of $H^2(\bulk,\bulk-\extr)$. Roughly speaking, we may think of $u$ as a map from two-dimensional chains with boundaries away from $\extr$, which counts, with signs, the number of intersections with $\extr$.

Now if $\bulk$ is oriented, the orientation of the normal bundle of $\extr$ induces also an orientation of the tangent bundle of $\extr$. This can be characterized by a relative homology class $\extrhom\in H_{d-1}(\extr,\entsurf)$.

Consider now the two cohomological long exact sequences, of the pair $(\bulk-\extr,\bdy-\entsurf)$, and of the triple $(\bulk,\bulk-\extr,\bdy-\entsurf)$, as explained in \eqref{eq:cohomologyTriple}, which fit into a diagram:
\begin{equation*}
\kern-2em
\begin{tikzcd}[row sep=small,column sep=small]
\arrow{dr} & & H^1(\bulk-\extr) \arrow{r}{i^*} \arrow{dd}{\delta} & H^1(\bdy-\entsurf) \arrow{dr}{\delta} \arrow{dd}{\delta} & & \hbox{} \\
& H^1(\bulk-\extr,\bdy-\entsurf) \arrow{ur}{j^*} \arrow{dr}{\delta} & & & H^2(\bulk-\extr,\bdy-\entsurf) \arrow{ur} \arrow{dr} & \\
\arrow{ur} & & H^2(\bulk,\bulk-\extr) \arrow{r}{j^*} & H^2(\bulk,\bdy-\entsurf) \arrow{ur}{i^*} & & \hbox{}
\end{tikzcd}\
\end{equation*}
where the $i^*$ and $j^*$ maps are the relevant restrictions and extensions (dual to inclusion and quotient) respectively, and all the $\delta$s are various relative coboundary maps. The top and bottom rows are the usual exact sequences of relative cohomology, and it is straightforward to show that the diagram commutes at the level of cochains. It is then an easy exercise\footnote{ Exercise 38 in chapter 2.2 of \cite{Hatcher:2002}.} to show that this induces another long exact sequence, the crucial part of which is
\begin{equation*}
\begin{tikzcd}
H^1(\bulk-\extr) \arrow{r}{{(i^*,\delta)}} & H^1(\bdy-\entsurf)\oplus H^2(\bulk,\bulk-\extr) \arrow{r}{\delta - j^*}& H^2(\bulk,\bdy-\entsurf)\end{tikzcd}
\end{equation*}
where the maps are the obvious ones from the diagram above, except that one factor in the second map has a minus sign.

Now an element $\psi\in H^1(\bulk-\entsurf)$ restricting to $\phi$ on the boundary means that
$i^*\psi=\phi$, and restricting to $\psi_F$ on the fibres is equivalent to $\delta \psi=u$. So $\psi$ with the desired properties exists iff $\phi\oplus u$ is in the image of the first map, which by exactness is the kernel of the second, so this is equivalent to $\delta \phi = j^* u$.

Finally, there is a generalization of Poincar\'e-Lefschetz duality \eqref{eq:PoincareLefschetz} that by splitting $\bdy$ into a tubular neighborhood of $\entsurf$ and its complement gives us an isomorphism $H^2(\bulk,\bdy-\entsurf)\approx H_{d-1}(\bulk, \entsurf)$, via cap product with the fundamental class $\mu\in H_{d+1}(\bulk,\bdy)$. The theorem follows by dualizing the equality $\delta \phi = j^* u$ under this, and showing that the duals of $\delta \phi$ and $j^* u$ are repectively the inclusions of $\reghom$ and $\extrhom$ into $H_{d-1}(\bulk,\bdy-\entsurf)$.

The dual of $\delta \phi$ is $\mu\frown\delta\phi$, which equals the inclusion of $\partial\mu\frown\phi$ into the bulk. This is indeed the correct thing, since $\partial\mu\frown\phi$ is just the boundary Poincar\'e-Lefschetz duality we originally used to relate $\reghom$ and $\phi$.

The final part, computing the dual of $j^* u$, generalizes a well-known result in the case of closed manifolds \cite{samelson1965}. The dual of $j^* u$ is $\mu\frown j^* u=j_*\mu\frown u$, where $j_*\mu\in H_{d+1}(\bulk,(\bulk-\entsurf)\cup\bdy)$. By restricting this cap product to a tubular neighborhood $N$ of $\entsurf$, with inclusion map $k:N\to\bulk$, it is clear that the result must be in the image of $k_*$. Choosing $N$ such that it retracts onto $\extr$ by a retraction $r:N\to \extr$ homotopic to the identity, it further follows that $k_*=i_* r_*$, so we get a cycle in the image of $i_*$. Working separately in each connected component of $\extr$, this must be $n \, i_*\extrhom$ for some integer $n$. A local analysis in a trivialization of $N$, where the computations are simply in $\mathbb{R}^{d+1}$, shows that $n=1$, so the dual of $j^* u$ must in fact equal $i_*\extrhom$ as required.

\end{proof}

\subsection{Non-orientable manifolds}
\label{sec:NonOrientable}

In the case when the bulk is not orientable, the cohomological parts of the above proof go through unchanged, but the dualities used to make statements in terms of homology are not applicable. In fact, since $\extr$, for instance, need not be orientable, there may not even be a homology class in $H_{d-1}(\bulk,\entsurf;\mathbb{Z})$ that represents it, so it becomes less clear how to state the result. This difficulty can be overcome by passing to $\mathbb{Z}_2$ coefficients, but since the correct cohomological argument requires integer coefficients, this gives an insufficiently strong constraint on allowed $\extr$. The correct thing is instead to use homology with local (or twisted) coefficients (see \cite{Hatcher:2002}, section 3.H).

The generalization to local coefficients is performed by, roughly speaking, allowing the coefficients of the chain complexes to live in some module of the fundamental group of the spacetime. In our case, this will be the orientation class, being the module with action given by $\pm 1$ depending on whether traversing a loop preserves or reverses orientation.

There is a generalized Poincar\'e-Lefschetz duality that holds for non-orientable spacetimes, between integer cohomology, and homology with coefficients in this module twisted by the orientation class. Since the cohomology we use is ambivalent to the presence or absence of orientation, nothing is altered until the very end, when we take this duality, and as such the homology statement will be in terms of local coefficients.

\section{Discussion}
\label{sec:discussion}

The generalized gravitational entropy construction of LM, which is inspired by the replica trick in field theory,
provides a derivation of the RT prescription for holographic  entanglement entropy. We examined the conditions under which the LM construction guarantees, at the level of topology, that the extremal surface computing entanglement entropy satisfies the homology constraint in the bulk.

Suppose we have a bulk Euclidean geometry $\bulk$, a boundary region $\regA$ and a minimal surface $\extr$, a candidate for computing the entanglement entropy of $\regA$. The local, dynamical part of the LM argument relates this to the replica trick as follows. Introduce a conical deficit at $\extr$, with defect angle $\frac{2\pi}{q}$, increasing $q$ away from unity, whilst changing the geometry to keep it on shell away from $\extr$. Locally, near $\extr$, we may choose a `cigar' where there are polar coordinates parametrising a distance from $\extr$ and an angle around it, with $\extr$ itself at the origin. When $q$ is an integer, this cigar with a conical deficit can be unwrapped by allowing the period of the angle to become $q$ times larger, so we obtain a smooth space. The resulting boundary is \emph{locally} the correct space with which to compute the $q^{\rm th}$ R\'enyi entropy, which makes contact with the replica trick.

It is less obvious that this picture works globally, since there is usually not a globally valid `cigar', with a circle parameterised by a global `Euclidean time' shrinking to zero size at $\extr$. There is no obvious obstruction to introducing a conical defect at any minimal surface to give a continuous family of geometries, but at integer $q$, when the defects are locally unwrapped, it may not be possible to consistently extend the local pictures to the whole spacetime, or doing say may produce the R\'enyi entropy for the wrong region. For example, on a torus, the replicated geometries for an interval and its complement are locally the same, but differ globally by what happens on traversing the nontrivial cycles of the torus. This global picture, left implicit in LM, is crucial for understanding which choices for $\extr$ give rise to the correct R\'enyi entropies.

Our main result addresses this global aspect, and can be phrased as follows: \emph{A conical deficit introduced at a given extremal surface $\extr$ comes from a quotient of a geometry whose boundary is the replicated space relevant for the $q$th R\'enyi entropy of region $\regA$, for every positive integer $q$, if and only if $\extr$ is homologous to $\regA$}.  In particular, using the terminology explained in the introduction, the local interpretation of the LM assumptions suffices to enforce the homology constraint so long as it holds at all $q$ and is consistent with the replica trick at integer $q$.

In this statement it is crucial that the construction of the bulk covering space works at all integer values of $q$. We gave an explicit example to show that branching along homology violating extremal surfaces can
occur when the replica symmetric covering space exists for some value of $q$, but not for others. For the   cosmic brane construction of LM \cite{Lewkowycz:2013nqa} which starts from a singular spacetime $\bulk_q$,  the homology constraint is implemented in the $q\rightarrow 1$ limit, iff for all integer $q$, $\bulk_q$ lifts to a non-singular $q$-fold replica symmetric branched cover (asymptoting to the replicated boundary geometry).  We emphasize that -- since we have worked only at the level of topology -- these $\bulk_q$ need not be actual solutions to the theory but merely smooth manifolds satisfying appropriate boundary conditions at the cosmic brane and at infinity. The dynamical part of the argument, which requires the geometries to be on-shell (at least in a neighbourhood of $q=1$), is essentially independent of our considerations.

As such our analysis requires that the entanglement entropy is computed by an extremal surface which arises within a family of cosmic brane configurations as a limiting case. We do not require that the R\'enyi entropies themselves for integer values of $q$ are computed by the elements of the same family. They could be computed (at large central charge) by  other families, which would allow for R\'enyi phase transitions as in \cite{Belin:2013dva}, or even isolated configurations. This also implies that the logarithmic negativity, which is obtained for bi-partitioning of a pure state by analytic continuation of the even-$q$ R\'enyi entropies \cite{Calabrese:2012nk}, could arise from a completely different family, for we only need a family of cosmic brane solutions for even integral $q$.

We emphasize once more that the arguments were purely topological, with no recourse to dynamical information of the bulk gravitational  theory. In particular, they hold for any choice of the boundary region $\regA$ including disconnected ones. The statements regarding topological consistency, covering spaces, and the proof of
Theorem \ref{thm:homology} are applicable to such situations. Of course, being agnostic of the dynamics has drawbacks, in that we will be unable to decide which of the homology respecting extremal surfaces actually computes the entanglement entropy. On the other hand, this gives us the advantage of being able to allow breaking replica symmetry away from integral values of $q$, as perhaps may be necessary for non Einstein-Hilbert gravitational dynamics \cite{Camps:2014voa}.

We have assumed  the bulk spacetime to be orientable in our discussion. For non-orientable spacetimes
as indicated in \S\ref{sec:NonOrientable}  the natural generalization involves  homology groups
with twisted coefficients as opposed to integral homology used in our proof. This is relevant, for example, in the case of the $\mathbb{RP}^2$-geon spacetime \cite{Louko:1998hc}, which is a pure state of a two-dimensional CFT (on a Klein bottle). In the bulk the spacetime has a horizon but a single asymptotic boundary
(hence single-exterior black hole).  In this case the homology condition is necessary to argue that the entanglement entropy for a region and its complement (on a spatial circle of the boundary) are the same as required by purity; see \cite{Maxfield:2014kra} for an analysis of the entanglement structure in the
$\mathbb{RP}^2$-geon (on the $t=0$ slice the ${\mathbb Z}_2$-valued homology group will suffice).
It would be interesting to flesh out the details of the twisted homology constraint for more general physical examples.

Finally, let us also note that our analysis takes seriously the Euclidean LM construction. We only require the boundary region $\regA$ to lie on the time-reflection symmetric surface. We are not a-priori guaranteed then that the extremal surface $\extr$, and consequentially the homology surface $\homsurfA$, lie at the moment of time-reflection symmetry in the bulk (see \cite{Maxfield:2014kra} for an example in thermal AdS where a sub-dominant extremal surface of this kind was described). If such surfaces give the dominant contribution to the holographic entanglement entropy, then various results in the holographic entanglement entropy literature  such as strong-subadditivity \cite{Headrick:2007km,Wall:2012uf} and causality \cite{Wall:2012uf,Headrick:2014cta}  would have to be revisited.

\acknowledgments
We would like to thank  Matthew Headrick, Veronika Hubeny and Aitor Lewkowycz for illuminating discussions.
DM and MR would like to acknowledge the hospitality of the Issac Newton Institute, Cambridge where this project was initiated.
TH, DM and MR would also like to thank the Aspen Center for Physics (supported by the National Science Foundation under Grant 1066293) for their hospitality.

FH is supported by a Durham Doctoral Fellowship. DM was supported in part by the U.S. National Science Foundation under Grant No PHY11-25915, and by funds from the University of California. HM is supported by a STFC studentship. MR was supported in part by the FQXi  grant ``Measures of Holographic Information" (FQXi-RFP3-1334), the STFC Consolidated Grant ST/L000407/1, and by the ERC Consolidator Grant  ERC-2013-CoG-615443: SPiN (Symmetry Principles in Nature).

\appendix

\section{Equivalence of topological consistency conditions}
\label{sec:theoremsproofs}

We here prove two theorems which will make precise the equivalence between the existence of a replica symmetric covering space as in \eqref{eq:ConsCondBB}, and the existence of a sheet counting map in the bulk factoring the boundary sheet counting map by the inclusion \eqref{eq:ConsCondAA}.

To make the structure of the arguments clearer, we will prove something more general. We will take $X$ to be a subspace of $Y$, and look for covering spaces $\tilde{X}$ and $\tilde{Y}$, with an inclusion map consistent with the inclusion of the base spaces. We additionally require the covers to possess certain symmetries, specifically that they are normal with deck transformation groups $G$ and $H$ respectively. The existence of such a cover will be shown to be equivalent to the existence of a commutative square of homomorphisms between $\pi_1(X)$, $\pi_1(Y)$, $G$ and $H$. See section 1.3 of \cite{Hatcher:2002} for the necessary mathematical material.

To apply this in the context we require it, identify $X$ and $Y$ with boundary and bulk geometries respectively, both with the branching surfaces removed. Then $G$ and $H$ are both taken as $\mathbb{Z}_q$, so that we are requiring the covers to be replica symmetric. The square of homomorphisms will then be equivalent to \eqref{eq:ConsCondAA}.

The first part of the proof, Theorem \ref{thm:LiftImpliesLM}, constructs the covering spaces given the appropriate homomorphisms. Theorem \ref{thm:LMimpliesLift} shows the converse, that the appropriate homomorphisms exist starting from the covering spaces.

\begin{theorem}\label{thm:LiftImpliesLM}
 Let $X$ and $Y$ be path-connected, locally path-connected, semi-locally simply-connected spaces, and $i:X\to Y$ an injective inclusion map. Let $\phi:\pi_1(X)\to G$ and $\psi:\pi_1(Y)\to H$ be surjective homomorphisms, and $\rho:G\to H$ an injective homomorphism such that
\begin{equation*}
\begin{tikzcd}
\pi_1(X) \arrow{r}{\phi} \arrow{d}{i_*}
&G \arrow{d}{\rho}\\
\pi_1(Y) \arrow{r}{\psi} &H
\end{tikzcd}
\end{equation*}
commutes. Then there are covering spaces $p:\tilde X\to X$, and $P:\tilde Y \to Y$, with deck transformation groups $G,H$ and an injective inclusion map $\tilde i:\tilde X \to \tilde Y$ such that
\begin{equation*}
\begin{tikzcd}
\tilde X \arrow{r}{p} \arrow{d}{\tilde{i}}
&X \arrow{d}{i}\\
\tilde Y \arrow{r}{P} &Y
\end{tikzcd}
\end{equation*}
commutes.
\end{theorem}

In the case when $G=H$($=\mathbb{Z}_q$, say), and the groups are finite, $\rho$ is an isomorphism so we can identify the two groups. This is the way we use the theorem above. A similar remark applies for the second theorem.

\begin{proof}
Denote homotopy equivalence classes of loops by $[\cdot]$, and reversal of curves by $\bar\cdot$.

We construct $\tilde X$ as the set of equivalence classes $[\gamma]_{\tilde X}$ of curves $\gamma:[0,1]\to X$ starting at the basepoint in $X$, where $\gamma\sim_{\tilde X}\eta$ if (i) $\gamma(1)=\eta(1)$ and (ii) $[\gamma\bar\eta]\in\ker(\phi)$. The projection map is $p:[\gamma]_{\tilde X}\mapsto \gamma(1)$, which is clearly well-defined. This is the standard construction of the covering space with fundamental group $\ker(\phi)$. The deck transformation group is $\pi_1(X)/\ker(\phi)\cong G$, since $\phi$ is surjective. Similar statements apply to construct $\tilde Y$.

It remains only to construct the inclusion map $\tilde{i}$, by $\tilde{i}:[\gamma]_{\tilde X}\mapsto [i\gamma]_{\tilde Y}$. This is well defined, since if $\gamma\sim_{\tilde X}\eta$, then (i) $i\gamma(1)=i\eta(1)$, and (ii) $\psi[i\gamma\overline{i\eta}]=\psi i_*[\gamma\bar\eta]=\rho\phi[\gamma\bar\eta]=1$ since $[\gamma\bar\eta]\in\ker\phi$, so $[i\gamma\overline{i\eta}]\in\ker\psi$. Finally, $\tilde{i}$ is injective, since if $[i\gamma]_{\tilde Y}=[i\eta]_{\tilde Y}$, then (i) $i\gamma(1)=i\eta(1)\Rightarrow\gamma(1)=\eta(1)$ by injectivity of $i$, and (ii) $1=\psi[i\gamma\overline{i\eta}]=\psi{i}_*[\gamma\bar\eta]=\rho\phi[\gamma\bar\eta]\Rightarrow[\gamma\bar\eta]\in\ker(\phi)$, where we have used injectivity of $\rho$ in the last step.
\end{proof}

\begin{theorem}\label{thm:LMimpliesLift}
Suppose that there are normal covering spaces $p:\tilde X\to X$, and $P:\tilde Y \to Y$, with deck transformation groups $G,H$ and injective inclusion maps $i:X\to Y$ and $\tilde{i}:\tilde X \to \tilde Y$ such that
\begin{equation*}
\begin{tikzcd}
\tilde X \arrow{r}{p} \arrow{d}{\tilde{i}}
&X \arrow{d}{i}\\
\tilde Y \arrow{r}{P} &Y
\end{tikzcd}
\end{equation*}
commutes. Then an injective homomorphism $\rho:G\to H$ exists, such that the following diagram commutes:
\begin{equation*}
\begin{tikzcd}
\pi_1(\tilde X) \arrow{r}{p_*} \arrow{d}{\tilde{i}_*}
&\pi_1(X) \arrow{r}{\phi} \arrow{d}{i_*}
&G \arrow{r} \arrow[dashed]{d}{\rho}
&1\\
\pi_1(\tilde Y) \arrow{r}{P_*} &
\pi_1(Y) \arrow{r}{\psi} &
H \arrow{r}& 1
\end{tikzcd}
\end{equation*}
The maps $\phi$ and $\psi$ here are the quotient maps by the fundamental groups of the covering spaces (using the fact that the coverings are normal), which implies that the rows are exact.
\end{theorem}
\begin{proof}
Let $g\in G$. Since $\phi$ is surjective, there is some loop $\gamma$ in $X$ so that $\phi([\gamma])=g$. Define $\rho(g)=\psi i_*[\gamma]$. We must check that this is well defined, so suppose that $\eta$ is another loop in $X$ with $\phi([\eta])=g$. Then $\phi([\gamma\bar\eta])=\phi([\gamma][\eta]^{-1})=1$. This means that $\gamma\bar\eta$ lifts to a loop $\tilde\gamma$ in $\tilde X$, by exactness of the top row. Now $i_*[\gamma\bar\eta]=i_*p_*[\tilde\gamma]=P_*\tilde{i}_*[\tilde\gamma]\in\ker(\psi)$ by exactness of the bottom row, so $\psi i_*[\gamma]=\psi i_*[\eta]$. It's also clear from the definition that $\psi i_*=\rho\phi$.

Finally, we check injectivity of $\rho$, so let $g\in G$ be distinct from the identity. Pick a loop $\gamma$ with $\phi[\gamma]=g$, and lift it to a curve $\tilde\gamma$ in $\tilde X$. Since $g$ is not the identity, this is not a loop, and since $\tilde{i}$ is injective, the curve $\tilde{i}\tilde\gamma$ in $\tilde Y$ is also not a loop. By the commutativity, this is the same curve as the lift of $i\gamma$, which implies that $[i\gamma]\notin\im(P_*)=\ker(\psi)$. So $\rho(g)=\psi([i\gamma])$ is not the identity, and $\rho$ is injective.
\end{proof}

\section{Review of algebraic topology}
\label{sec:algebraictopology}

This appendix reviews the algebraic topology required for the main result of the paper. See \cite{Hatcher:2002} for more detailed discussions.

The space of singular $n$-chains $C_n(X)$ of a topological space $X$ is the free abelian group with basis singular $n$-simplices, which are maps from the standard $n$-simplex $\Delta^n=\{(t_0,\ldots,t_n)\in \mathbb{R}^{n+1}|\sum_i t_i =1, t_i\geq 0\}$ to $X$. The boundary maps $\partial_n:C_n(X)\to C_{n-1}(X)$ act on chains $\sigma$ by a sum of restrictions of $\sigma$ to its $n+1$ faces, with apropriate signs. A cycle is defined as a chain with zero boundary, in the kernel of $\partial$, and a boundary is a chain in the image of $\partial$. Since $\partial_n\partial_{n+1}=0$, all boundaries are cycles, and this means we can define the singular homology groups as cycles modulo boundaries: $H_n(X)=\ker \partial_n / \im \partial_{n+1}$.

A generalization of this is a relative homology group, where we ignore what goes on in some part of the space. Let $A$ be a subset of the space $X$; the space of relative chains is the space of chains in $X$ modulo chains in $A$, $C_n(X,A)=C_n(X)/C_n(A)$. The usual boundary map continues to be well-defined between relative chains, so we may again take cycles modulo boundaries to get the relative homology groups $H_n(X,A)$. Cycles here may therefore have a boundary in the subspace $A$, and it is trivial if it is homologous to a cycle completely within $A$. Often the space $A$ will be the complement of some set, $A=X-B$, when $H_n(X,X-B)$ is sometimes referred to as `homology at $B$', since it ignores what happens to chains except at the subspace $B$.

Since the relative homology ignores what goes on in some subspace, nothing is lost by removing some part of that subspace. This is made precise by the excision theorem, which states that if $Z\subseteq A\subseteq X$, and the closure of $Z$ is contained in the interior of $A$, then $H_n(X,A)$ is isomorphic to $H_n(X-Z,A-Z)$.

Of central importance is the fact that relative homology groups fit into a long exact sequence:
\begin{equation}
\cdots\xrightarrow{\partial}H_n(A)\xrightarrow{i_*}H_n(X)\xrightarrow{j_*}H_n(X,A)\xrightarrow{\partial}H_{n-1}(A)\xrightarrow{i_*}
\cdots
\end{equation}
The maps $i_*$ here are inclusions of cycles in $A$ into $X$, and the maps $j_*$ are the quotients by chains in $A$ to get to relative homology. Finally, the maps $\partial$ are boundary maps: any relative $n$-cycle must have boundary contained in $A$ by definition, and this gives a homology class in $H_{n-1}(A)$. Exactness is geometrically very intuitive: for example, a cycle in $H_n(A)$ gives zero when it is included into $H_n(X)$ if and only if it is the boundary of some chain in $X$, which is the case if and only if it is in the image of the boundary map $\partial$ from the relative homology group. So $\ker i_*=\im \partial$.

This can be slightly generalized to the long exact sequence of the triple $(X,A,B)$, where $B\subseteq A\subseteq X$, by doing everything relative to the smallest subspace $B$:
\begin{equation}
\cdots\xrightarrow{\partial}H_n(A,B)\xrightarrow{i_*}H_n(X,B)\xrightarrow{j_*}H_n(X,A)\xrightarrow{\partial}H_{n-1}(A,B)\xrightarrow{i_*}\cdots
\end{equation}

Now most of the work we will do will be not in terms of homology, but cohomology, to which all of the above carries over. To pass to cohomology, all the relevant spaces and maps should be dualized: we consider spaces of cochains $C^n(X)=\hom(C_n(X),\mathbb{Z})$, joined by coboundary maps $\delta_n:C^n(X)\to C^{n+1}(X)$ dual to the boundary maps, so $\delta f (x)=f(\partial x)$. The kernel of $\delta$ gives the cochains, and the image the coboundaries; the cohomology groups are cochains mod coboundaries $H^n(X)=\ker \delta_n / \im \delta_{n-1}$.

The constructions for relative homology groups carry over to cohomology, by dualizing the relative chain complex. In particular the excision theorem is the same in cohomology, and the long exact sequences of relative homology groups dualize, for example for the triple $(X,A,B)$ we have
\begin{equation}\label{eq:cohomologyTriple}
\cdots\xrightarrow{i^*}H^{n-1}(A,B)\xrightarrow{\delta}H^n(X,A)\xrightarrow{j^*}H^n(X,B)\xrightarrow{i^*}H^{n}(A,B)\xrightarrow{\delta}\cdots
\end{equation}

While most of the argument is phrased in terms of cohomology, in the end we want to translate the result into homology language. The crucial tool to do this is a generalization of Poincar\'e duality. Geometrically, these dualities can be roughly thought of as taking a submanifold to a function on submanifolds of complementary dimension, counting with signs the number of intersections between the submanifolds.

This requires the cap product between chains and cochains, defined by
\begin{gather}
 \frown: \, C_k(X) \times C^l (X) \longrightarrow C_{k-l}(X) \\
 \sigma \frown \varphi = \varphi(\sigma_{[v_0,\ldots,v_l]}) \cdot \sigma_{[v_l,\ldots v_k]}\nonumber
\end{gather}
where $k\geq l$ and $\sigma_{[v_0,\ldots, v_l]}$ denotes the restriction of $\sigma$ to the simplex spanned by the vertices $v_0,\ldots,v_l$. This map induces a cap product between homology and cohomology classes.

A connected closed orientable $n$-manifold $X$ has top homology group $H_n(X)$ isomorphic to $\mathbb{Z}$, and an orientation of $X$ is equivalent to a generator of this group, a \emph{fundamental class} $[X]$. The Poincar\'e duality map is given by the cap product with the fundamental class $[X]\frown:H^k(X)\to H_{n-k}(X)$, which is in fact an isomorphism between complementary homology and cohomology groups.

This duality generalizes to manifolds with boundary. Crucial to this is a relative version of the cap product. Using the same definitions as before, it can be checked that the product on chains and cochains induces a product on relative (co)homology groups
\begin{equation}
 \frown: \, H_k(X,A\cup B) \times H^l (X,A) \longrightarrow H_{k-l}(X,B)
\end{equation}
where $A$ and $B$ are open sets in $X$. The duality theorem we need applies this in the case of an oriented $n$-manifold $X$ with boundary $\partial X$, where the boundary is decomposed as $\partial X=A\cup B$, where $A$ and $B$ are $(n-1)$-dimensional manifolds with common boundary $A\cap B=\partial A=\partial B$. An orientation of $X$ is defined in this case by a generator of $H_n(X,\partial X)$, which as before gives us isomorphisms
\begin{equation}\label{eq:PoincareLefschetz}
[X]\frown:H^k(X,A)\xrightarrow{\approx} H_{n-k}(X,B).
\end{equation}
This includes the special cases where $A=\emptyset$ and $B=\partial X$, or vice-versa, known as Poincar\'e-Lefschetz duality.


\begin{thebibliography}{10}

\bibitem{Ryu:2006bv}
S.~Ryu and T.~Takayanagi, {\it {Holographic derivation of entanglement entropy
  from AdS/CFT}},  {\em Phys.Rev.Lett.} {\bf 96} (2006) 181602,
  [\href{http://xxx.lanl.gov/abs/hep-th/0603001}{{\tt hep-th/0603001}}].

\bibitem{Ryu:2006ef}
S.~Ryu and T.~Takayanagi, {\it {Aspects of Holographic Entanglement Entropy}},
  {\em JHEP} {\bf 0608} (2006) 045,
  [\href{http://xxx.lanl.gov/abs/hep-th/0605073}{{\tt hep-th/0605073}}].

\bibitem{Hubeny:2007xt}
V.~E. Hubeny, M.~Rangamani, and T.~Takayanagi, {\it {A Covariant holographic
  entanglement entropy proposal}},  {\em JHEP} {\bf 0707} (2007) 062,
  [\href{http://xxx.lanl.gov/abs/0705.0016}{{\tt arXiv:0705.0016}}].

\bibitem{Lewkowycz:2013nqa}
A.~Lewkowycz and J.~Maldacena, {\it {Generalized gravitational entropy}},  {\em
  JHEP} {\bf 1308} (2013) 090, [\href{http://xxx.lanl.gov/abs/1304.4926}{{\tt
  arXiv:1304.4926}}].

\bibitem{Fursaev:2006ih}
D.~V. Fursaev, {\it {Proof of the holographic formula for entanglement
  entropy}},  {\em JHEP} {\bf 0609} (2006) 018,
  [\href{http://xxx.lanl.gov/abs/hep-th/0606184}{{\tt hep-th/0606184}}].

\bibitem{Headrick:2007km}
M.~Headrick and T.~Takayanagi, {\it {A Holographic proof of the strong
  subadditivity of entanglement entropy}},  {\em Phys.Rev.} {\bf D76} (2007)
  106013, [\href{http://xxx.lanl.gov/abs/0704.3719}{{\tt arXiv:0704.3719}}].

\bibitem{Headrick:2013zda}
M.~Headrick, {\it {General properties of holographic entanglement entropy}},
  {\em JHEP} {\bf 1403} (2014) 085,
  [\href{http://xxx.lanl.gov/abs/1312.6717}{{\tt arXiv:1312.6717}}].

\bibitem{Headrick:2014cta}
M.~Headrick, V.~E. Hubeny, A.~Lawrence, and M.~Rangamani, {\it {Causality \&
  holographic entanglement entropy}},
  \href{http://xxx.lanl.gov/abs/1408.6300}{{\tt arXiv:1408.6300}}.

\bibitem{Dong:2013qoa}
X.~Dong, {\it {Holographic Entanglement Entropy for General Higher Derivative
  Gravity}},  {\em JHEP} {\bf 1401} (2014) 044,
  [\href{http://xxx.lanl.gov/abs/1310.5713}{{\tt arXiv:1310.5713}}].

\bibitem{Camps:2013zua}
J.~Camps, {\it {Generalized entropy and higher derivative Gravity}},  {\em
  JHEP} {\bf 1403} (2014) 070, [\href{http://xxx.lanl.gov/abs/1310.6659}{{\tt
  arXiv:1310.6659}}].

\bibitem{Wall:2012uf}
A.~C. Wall, {\it {Maximin Surfaces, and the Strong Subadditivity of the
  Covariant Holographic Entanglement Entropy}},
  \href{http://xxx.lanl.gov/abs/1211.3494}{{\tt arXiv:1211.3494}}.

\bibitem{Faulkner:2013yia}
T.~Faulkner, {\it {The Entanglement Renyi Entropies of Disjoint Intervals in
  AdS/CFT}},  \href{http://xxx.lanl.gov/abs/1303.7221}{{\tt arXiv:1303.7221}}.

\bibitem{Barrella:2013wja}
T.~Barrella, X.~Dong, S.~A. Hartnoll, and V.~L. Martin, {\it {Holographic
  entanglement beyond classical gravity}},  {\em JHEP} {\bf 1309} (2013) 109,
  [\href{http://xxx.lanl.gov/abs/1306.4682}{{\tt arXiv:1306.4682}}].

\bibitem{Headrick:2012fk}
M.~Headrick, A.~Lawrence, and M.~Roberts, {\it {Bose-Fermi duality and
  entanglement entropies}},  {\em J.Stat.Mech.} {\bf 1302} (2013) P02022,
  [\href{http://xxx.lanl.gov/abs/1209.2428}{{\tt arXiv:1209.2428}}].

\bibitem{Witten:1999xp}
E.~Witten and S.-T. Yau, {\it {Connectedness of the boundary in the AdS / CFT
  correspondence}},  {\em Adv.Theor.Math.Phys.} {\bf 3} (1999) 1635--1655,
  [\href{http://xxx.lanl.gov/abs/hep-th/9910245}{{\tt hep-th/9910245}}].

\bibitem{Maldacena:2004rf}
J.~M. Maldacena and L.~Maoz, {\it {Wormholes in AdS}},  {\em JHEP} {\bf 0402}
  (2004) 053, [\href{http://xxx.lanl.gov/abs/hep-th/0401024}{{\tt
  hep-th/0401024}}].

\bibitem{Hatcher:2002}
A.~Hatcher, {\em Algebraic Topology}.
\newblock Cambridge University Press, 2002.

\bibitem{samelson1965}
H.~Samelson, {\it {On the Thom class of a submanifold}},  {\em Michigan Math.
  J.} {\bf 12} (09, 1965) 257--261.

\bibitem{Belin:2013dva}
A.~Belin, A.~Maloney, and S.~Matsuura, {\it {Holographic Phases of Renyi
  Entropies}},  {\em JHEP} {\bf 1312} (2013) 050,
  [\href{http://xxx.lanl.gov/abs/1306.2640}{{\tt arXiv:1306.2640}}].

\bibitem{Calabrese:2012nk}
P.~Calabrese, J.~Cardy, and E.~Tonni, {\it {Entanglement negativity in extended
  systems: A field theoretical approach}},  {\em J.Stat.Mech.} {\bf 1302}
  (2013) P02008, [\href{http://xxx.lanl.gov/abs/1210.5359}{{\tt
  arXiv:1210.5359}}].

\bibitem{Camps:2014voa}
J.~Camps and W.~R. Kelly, {\it {Generalized gravitational entropy without
  replica symmetry}},  \href{http://xxx.lanl.gov/abs/1412.4093}{{\tt
  arXiv:1412.4093}}.

\bibitem{Louko:1998hc}
J.~Louko and D.~Marolf, {\it {Single exterior black holes and the AdS / CFT
  conjecture}},  {\em Phys.Rev.} {\bf D59} (1999) 066002,
  [\href{http://xxx.lanl.gov/abs/hep-th/9808081}{{\tt hep-th/9808081}}].

\bibitem{Maxfield:2014kra}
H.~Maxfield, {\it {Entanglement entropy in three dimensional gravity}},
  \href{http://xxx.lanl.gov/abs/1412.0687}{{\tt arXiv:1412.0687}}.

\end{thebibliography}
\providecommand{\href}[2]{#2}\begingroup\raggedright\endgroup

\end{document}